\documentclass[12pt]{article}

\usepackage[english]{babel}
\usepackage[utf8]{inputenc}
\usepackage[authoryear]{natbib}
\usepackage[margin=1in]{geometry}
\usepackage{setspace}

\usepackage{amsmath}
\usepackage{amsfonts}
\usepackage{amsthm}
\usepackage{mathrsfs}
\usepackage{amstext}
\usepackage{comment}
\usepackage{epsfig}
\usepackage{latexsym}
\usepackage{epstopdf}
\usepackage{bm}
\usepackage{algorithm}
\usepackage{algpseudocode}
\usepackage{amssymb}
\usepackage{url}
\usepackage[colorlinks = true,linkcolor = blue]{hyperref}
\hypersetup{
        urlcolor = blue,
	colorlinks = true,
	linkcolor = blue,
	citecolor = blue,
}
\usepackage{enumitem}

\usepackage{multirow}
\usepackage{multicol}
\usepackage{babel}
\usepackage{array}
\usepackage{rotating}
\usepackage{longtable}
\usepackage{float}
\usepackage{booktabs}
\usepackage{lscape}
\usepackage[labelformat=default]{caption}

\newtheorem{theorem}{Theorem}
\newtheorem{remark}{Remark}
\newtheorem{assumption}{Assumption}
\newtheorem{lemma}{Lemma}

\newtheorem{proposition}{Proposition}

\newcommand{\E}{\mathbb{E}}
\newcommand{\Var}{\mathrm{Var}}

\newcommand{\xhat}{\widehat{\Theta}}
\newcommand{\xtil}{\widetilde{\Theta}}
\newcommand{\tilY}{\widetilde{Y}}
\newcommand{\sumNc}{\sum_{i=1}^{N_c}}
\newcommand{\sumC}{\sum_{c=1}^C}
\newcommand{\Dcij}{\Delta_{ij\mid c}}
\newcommand{\Dcijk}{\Delta_{ijk\mid c}}
\newcommand{\Dcijkl}{\Delta_{ijkl\mid c}}

\title{The Exact Variance of the Average Treatment Effect Estimator in Cluster Randomized Controlled Trials}

\author{Yue Fang\thanks{School of Management and Economics, The Chinese University of Hong Kong, Shenzhen. Email: \url{fangyue@cuhk.edu.cn.}} \quad Geert Ridder\thanks{Department of Economics, University of Southern California. Email: \url{ridder@usc.edu.}}}

\begin{document}

\maketitle
\begin{abstract}
    In cluster randomized controlled trials (CRCT) with a finite populations, the exact design-based variance of the Horvitz-Thompson (HT) estimator for the average treatment effect (ATE) depends on the joint distribution of unobserved cluster-aggregated potential outcomes and is therefore not point-identifiable. We study a common two-stage sampling design-random sampling of clusters followed by sampling units within sampled clusters-with treatment assigned at the cluster level. First, we derive the exact (infeasible) design-based variance of the HT ATE estimator that accounts jointly for cluster- and unit-level sampling as well as random assignment. Second, extending \cite{aronow2014sharp}, we provide a sharp, attanable upper bound on that variance and propose a consistent estimator of the bound using only observed outcomes and known sampling/assignment probabilities. In simulations and an empirical application, confidence intervals based on our bound are valid and typically narrower than those based on cluster standard errors.
\end{abstract}

\section{Introduction} \label{sec:intro}
Cluster randomized controlled trials are prevalent in applied economics \citep{baranov2020maternal,attanasio2020estimating,dhar2022reshaping}. Populations are partitioned into natural groups—cities, villages, schools, classrooms—and randomization occurs at the cluster level while outcomes are measured on individuals within clusters. We take a finite‑population, design‑based perspective with two sources of randomness, following \citet{abadie2020sampling,abadie2023should}: (i) sampling of the empirical study from a finite target population and (ii) treatment assignment. We explicitly allow for two‑stage sampling: first sample clusters, then sample units within sampled clusters; treatment is assigned at the cluster level.

We analyze design‑based inference for the finite‑population ATE when exposure is assigned by cluster and data are collected via two‑stage sampling. The estimand is the average of unit‑level potential outcomes in the finite population. Under known sampling and assignment probabilities, the Horvitz–Thompson (HT) estimator is unbiased. However, its exact variance depends on cluster‑aggregated potential outcomes that are not jointly observed (the familiar ``cross‑arm covariance'' problem), so the exact variance cannot be consistently estimated without additional assumptions.

A finite‑population perspective is natural for cluster randomized experiments with two‑stage sampling. It has long roots in survey sampling and has become increasingly important in causal inference and program evaluation, where researchers seek design‑faithful uncertainty quantification that conditions on the realized finite population and the actual sampling/assignment mechanisms. In parallel, a growing line of work emphasizes finite‑sample, design‑based inference for experiments as implemented. For example, \cite{heckman2021using} develop worst‑case randomization tests for compromised or partially documented experiments, modeling experimenters who ``satisfice'' on covariate balance and delivering exact small‑sample procedures that are robust to rerandomization and post‑assignment transfers. Their framework motivates finite‑sample, design‑based inference when assignment protocols are only partially known, and it explicitly targets settings where conventional asymptotics can be misleading. More broadly, recent contributions clarify when design‑ versus sampling‑based uncertainty is appropriate \citep{abadie2020sampling}, develop theory for rerandomization and studentized randomization tests \citep{morgan2015rerandomization,li2018asymptotic,chung2016multivariate}, study experiments ``as implemented'' rather than ``as intended'' \citep[e.g.,][]{heckman2010analyzing}, and assess the reliability of asymptotics in small experimental samples \citep{young2019channeling}. Our paper contributes to this agenda by deriving exact finite‑population variance formulas and sharp (and conservative) variance estimators for cluster randomized trials with two‑stage sampling, and by establishing asymptotic normality under weak moment and heterogeneity conditions—supplying practical, design‑based uncertainty quantification tailored to the settings most common in field experiments.

In particular, we derive the exact design‑based variance of the HT estimator of the ATE that incorporates randomness from cluster sampling, unit sampling, and cluster‑level assignment. Relative to designs without sampling or without clustering, two‑stage sampling introduces identifiable within‑ and between‑cluster variance components. Extending the sharp‑bounds approach of \citet{aronow2014sharp} to two‑stage sampling with cluster assignment, we obtain a sharp, consistently estimable upper bound for the variance by applying the bounds to the joint distribution of the estimated cluster averages—thereby embedding both cluster‑ and unit‑sampling noise. We then construct a consistent estimator of this bound using only observed data and known inclusion/assignment probabilities. In simulations and in an application, Wald intervals that plug in our bound estimator are typically tighter than intervals based on clustered standard errors (Liang–Zeger), which can be overly conservative in this design, while maintaining nominal coverage.

Our paper is closely connected to the literature on cluster randomization and design-based inference. A large literature studies experiments in which treatment is assigned at the cluster level and outcomes are measured on units within clusters. Early work formalized design-based estimands under cluster assignment, and subsequent papers developed randomization-based variance expressions and estimators for cluster-level designs \citep{middleton2015unbiased, su2021model, schochet2022design, wang2024design}.  Much of this literature treats the sample as given, i.e. conditioning on the set of observed clusters and units and focuses on assignment-induced uncertainty. We differ by explicitly integrating two-stage sampling uncertainty, sampling clusters, then sampling units within sampled clusters, into the variance of the Horvitz-Thompson (HT) estimator for ATE. This adds variance components that persist even if assignment were held fixed, and it changes the object that a variance estimator must target.

Our analysis follows the finite-population, design-based perspective emphasized by \cite{abadie2020sampling, abadie2023should} and by the survey-sampling literature on two-stage designs \citep{ohlsson1989asymptotic, chauvet2020inference}. Those references clarify that sampling and assignment are distinct randomization mechanisms whose contributions to uncertainty should be kept separate. Relative to these, we specialize to cluster-level assignment and deliver (i) the exact design-based variance of the HT estimator that aggregates both stages of sampling ans treatment assignment, and (ii) a sharp, consistently estimable upper bound for that variance when it is not point-identified due to missing cluster-aggregated potential outcomes.

\cite{aronow2014sharp} show that when cross-potential-outcome terms are unobserved, the variance of common design-based estimators may not be point-identified but admits a sharp upper bound under the randomization. We extend that sharp-bounds logit to two-stage sampling with cluster-level treatment assignment: the unobserved terms are cluster-aggregated potential outcomes, and the feasible bound must respect both sampling inclusion probabilities and treatment probabilities, being applied to the distribution of the estimated quantities. Our consistent estimator targets this bound, and hence supports valid yet typically tighter intervals than generic clustered standard errors in our settings.

Moreover, \cite{bugni2022inference} analyze the cluster-randomized experiments from a superpopulation perspective in which potential outcomes and even sample size are random variables. As they note (Remark 3.4), finite-population precision gains stemming from observing a large share of the population cannot be consistently estimated because counterfactual potential outcomes are missing. We adopt the finite-population perspective and show how to recover a portion of precision gains via a sharp design-based variance bound, i.e., we trade point identification of the variance for a tight, attainable bound that is consistently estimable under the actual design.

The HT estimator weights observed outcomes by the inverse of the product of sampling and assignment probabilities. With cluster-level assignment, each sampled cluster reveals only one of its two cluster-aggregated potential outcomes, so the covariance term that involves the joint distribution of the potential outcomes in the variance cannot be pinned down. Our sharp bounds identifies the largest variance compatible with the design, the realized treatment assignments, and the observed outcomes. Because the bound respects the actual sampling and assignment probabilities, it is often materially less conservative than generic clustered standard errors. In practice, many applied papers report cluster standard errors computed from unit-level regressions, sometimes with or without survey weights. As emphasized by \cite{abadie2020sampling, abadie2023should} and others, such standard errors mix model-based approximation with design-based features and can be over- or under-conservative depending on the design, especially under unequal inclusion probabilities or when the estimand is finite-population.

Throughout the paper, we consider no interference across clusters, known sampling and assignment probabilities, possibly unequal across clusters, and treatment assignment at the cluster level (compliance issues are outside our scope). The sampling scheme is conducted without replacement at each stage and may allow probability-proportional-to-size draws.

The structure of this paper is as follows. Section \ref{subsec:set} formalizes the two-stage sampling and cluster-level assignment and introduce notation. Section \ref{sec:estimator} defines the HT estimator for the finite-population ATE and gives its exact design-based variance. Section \ref{sec:asym} develops asymptotic properties under increasing numbers of clusters in a sequence of finite populations. Section \ref{sec:bound}  extends \cite{aronow2014sharp} to derive a sharp variance bound and proposes a consistent estimator of the bound. Section \ref{sec:simulation} presents Monte Carlo evidence. Section \ref{sec:application} applies our methods to a school-level cluster randomized controlled trial and compares bound-based intervals to clustered standard errors. 

\section{Set up}\label{subsec:set}

We consider a finite population of size $N$, consisting of $C$ clusters. Cluster $c\in\{1,2,\dots,C\}$, contains $N_c$ units, and we take $N$, $C$, and all $N_c$ as fixed and known. This assumption is plausible in many field settings (e.g., development economics) where researchers typically gather population information from the study area, have census lists, or have access to administrative records. For asymptotic results, we consider a sequence of such finite populations, with the number of clusters $C$ going to infinity, allowing for heterogeneous cluster size $N_c$.

Randomness arises from both sampling and treatment assignment. The sampling process has two stages. Firstly, $S$ clusters are randomly selected from $C$ clusters without replacement. Secondly, in a sampled cluster $c$, $n_c$ units are sampled at random from $N_c$ units without replacement. The treatment assignment is independent of both sampling stages by design. The treatment is on the cluster level: $S_1$ out of the $S$ sampled clusters are randomly assigned to treatment without replacement, and the remaining $S_0=S-S_1$ clusters are assigned to control. All units within a given cluster receive the same treatment status. Only the sampled units are observed and used for analysis. 

In practice, the intervention is delivered to the entire cluster in some experiments, even though only a subset of units is sampled for measurement; and in some designs, the intervention is administered only to the sampled units. Our analysis formally accommodates both cases. However, we note that in the latter setting, extending conclusions to the entire cluster may involve extrapolation when considering potential spillover effects--issues that go beyond the scope of this paper. These considerations are related to the broad literature on two-stage randomization designs, which explicitly study within-cluster and between-cluster treatment allocations in the existence of spillover.

Let $Y_{ci}(1)$ and $Y_{ci}(0)$ denote the potential outcomes for unit $i$ in cluster $c$ under treatment and control, respectively. The potential outcomes are taken as fixed. We introduce several indicator variables to describe the sampling and assignment process.

Define $R_c \in \{0,1\}$ as the indicator variable denoting whether cluster $c$ is sampled ($R_c = 1$) or not ($R_c = 0$) in the first-stage sampling, and $D_c \in \{0,1\}$ as the treatment assignment indicator for cluster $c$, where $D_c = 1$ if $c$ is assigned to treatment and $D_c = 0$ if $c$ is assigned to control. Let $R_{i \mid c} \in \{0,1\}$ be the indicator of whether unit $i$ in cluster $c$ is sampled in the second-stage sampling, conditional on that cluster $c$ is sampled ($R_c = 1$) in the first stage. The overall sampling indicator for unit $i$ in cluster $c$ is then denoted as $R_{ci} = R_c R_{i \mid c}$.

The sampling and treatment assignment processes induce the following probabilities. The probability that cluster $c$ is sampled is
\[
p = \Pr(R_c = 1) = \frac{S}{C},
\]
where $S$ is the number of sampled clusters out of $C$ total clusters.

Within a sampled cluster $c$, the probability of unit $i$, which belongs to cluster $c$, being sampled is
\[
\pi_c = \Pr(R_{ci}=1\mid R_c=1) = \Pr(R_{i \mid c} = 1) = \frac{n_c}{N_c},
\]
where $n_c$ is the number of sampled units within cluster $c$ and $N_c$ is the total number of units in that cluster.

Conditional on being sampled, the probability that cluster $c$ is assigned to treatment is
\[
q = \Pr(D_c = 1 \mid R_c = 1) = \frac{S_1}{S},
\]
where $S_1$ clusters are assigned to treatment and $S_0=S-S_1$ clusters are assigned to control.

Because both sampling and treatment assignment are carried out without replacement, the inclusion events are dependent, affecting the joint inclusion probabilities of pairs of clusters and units. Specifically, the probability that both clusters $c$ and $c'$ ($c \neq c'$) are sampled is
\[
\Pr(R_c = 1, R_{c'} = 1) = \frac{S(S - 1)}{C(C - 1)}.
\]
Within cluster $c$, the probability that both units $i$ and $j$ ($i \neq j$) are sampled is
\[
\Pr(R_{i \mid c} = 1, R_{j \mid c} = 1) = \frac{n_c(n_c - 1)}{N_c(N_c - 1)}.
\]
For units in different clusters, the probability that unit $i$ in cluster $c$ and unit $j$ in cluster $c'$ are both sampled is
\[
\Pr(R_{i \mid c} = 1, R_{j \mid c'} = 1) = \frac{n_c n_{c'}}{N_c N_{c'}}.
\]
Conditional on both clusters $c$ and $c'$ being sampled, the probability that both clusters are assigned to treatment is
\[
\Pr(D_c = 1, D_{c'} = 1 \mid R_c = 1, R_{c'} = 1) = \frac{S_1(S_1 - 1)}{S(S - 1)},
\]
the probability that both clusters are assigned to control is 
\[
\Pr(D_c = 0, D_{c'} = 0 \mid R_c = 1, R_{c'} = 1) = \frac{S_0(S_0 - 1)}{S(S - 1)},
\]
and the probability that $c$ is assigned to treatment and $c'$ is assigned to control is
\[
\Pr(D_c = 1, D_{c'} = 0 \mid R_c = 1, R_{c'} = 1) = \frac{S_1S_0}{S(S - 1)}.
\]
In this two-stage design, the sampling process first determines which clusters and which units within them are observed, while treatment is assigned independently at the cluster level. The without-replacement nature of both sampling and assignment means that inclusion and treatment probabilities for different clusters or units are not independent. The above probabilities form the basis for the Horvitz-Thompson weighting scheme used in our estimator and for deriving its exact design-based variance.

\section{The Estimator of the ATE and its Exact Variance}\label{sec:estimator}

\subsection{The estimator}

The population parameter of interest is the average treatment effect (ATE), defined as
\begin{align}\label{ATE}
    \tau &= \frac{1}{N} \sum_{c=1}^C \sum_{i=1}^{N_c} \left( Y_{ci}(1) - Y_{ci}(0) \right),
\end{align}
where $Y_{ci}(1)$ and $Y_{ci}(0)$ denote the potential outcomes of unit $i$ in cluster $c$ under treatment and control, respectively. The potential outcomes are treated as fixed in the finite population, while all randomness arises from the sampling and treatment assignment in the experimental design.

We estimate $\tau$ using the Horvitz-Thompson (HT) estimator:
\begin{align*}\label{HT-tau}
    \widehat{\tau} &= \frac{1}{N} \sum_{c=1}^C \left[ \frac{R_c D_c}{p q} \sum_{i=1}^{N_c} \frac{R_{i \mid c} Y_{ci}}{\pi_c} - \frac{R_c (1 - D_c)}{p (1 - q)} \sum_{i=1}^{N_c} \frac{R_{i \mid c} Y_{ci}}{\pi_c} \right].
\end{align*}
This estimator is unbiased, because by construction $\Pr(R_cD_c=1)=pq$, $\Pr(R_c(1-D_c)=1)=p(1-q)$, and $\Pr(R_{i \mid c}=1)=\pi_c$.

\begin{remark}\label{rmk:diff-in-mean}
The difference-in-means estimator is
\[
\widehat{\tau}^{dm}
= \frac{1}{n_1}\sum_{c=1}^C\sum_{i=1}^{N_c} R_c R_{i|c} D_c Y_{ci}
 - \frac{1}{n_0}\sum_{c=1}^C\sum_{i=1}^{N_c} R_c R_{i|c} (1-D_c) Y_{ci},
\]
where
\[
n_1=\sum_{c,i} R_c R_{i|c} D_c, \qquad 
n_0=\sum_{c,i} R_c R_{i|c} (1-D_c)
\]
are the realized numbers of sampled treated and sampled control units. Thus, both denominators $n_1$ and $n_0$ are random.

By contrast, the HT estimator divides by the fixed population size $N$ and places all randomness in the sampling and treatment indicators $(R_c, R_{i|c}, D_c)$ inside the numerator. The randomness in the realized treatment and control sample sizes is already accounted for through these indicators. Because each observed outcome is weighted by inverse inclusion and treatment probabilities, the HT estimator reconstructs the population total directly, without introducing additional randomness through a random denominator. 
\end{remark}

Our goal in this section is to calculate the exact variance of $\widehat{\tau}$ and then propose a feasible and consistent estimator of that variance for inference purposes. To simplify notation and account for heterogeneity in cluster sizes, define the average cluster size
\[
\Bar{N} = \frac{1}{C} \sum_{c=1}^C N_c,
\]
and the average within-cluster sample size
\[
\Bar{n} = \frac{1}{C} \sum_{c=1}^C n_c.
\]
The quantity $N = C \Bar{N}$ represents the total number of units in the population. In contrast, $n = C \Bar{n}$ does not correspond to the actual number of sampled units in the realized data, because only $S$ out of $C$ clusters are drawn in the first stage. Instead, it represents the design-level number of sampling opportunities, which is the total number of units that would be sampled if all clusters were included in the first stage. The actual number of sampled units is random and equals $\sum_{c=1}^CR_cn_c$, which depends on which clusters are sampled, i.e., the realizations of $R_c$. This distinction is important because $\Bar{n}$ and $n = C \Bar{n}$ are fixed by design, not random. They are specified in advance to characterize the sampling probabilities and to simplify expressions of the estimator's variance. In many applications, especially in development economics, researchers know all cluster sizes $N_c$ from census or administrative data, so $\Bar{N}$ is treated as a known quantity. Also, in practice, the within-cluster sample sizes $n_c$ are almost always determined before the first-stage cluster sampling is carried out. Typically, one of the following conventions is adopted at the design stage: fixed number of sampled units per cluster $n_c=\Bar{n}$; fixed sampling proportion $n_c = \pi N_c$, where $\pi$ is a constant sampling rate applied across clusters; predetermined list of sample sizes, where a set of values $n_c$ are chosen in advance based on available information about cluster sizes, survey logistics, or budget constraints. Because these choices are made ex-ante, before any cluster is actually sampled, it is natural and appropriate to treat $n_c$, $\Bar{n}$, and $n = C\Bar{n}$ as known, nonrandom design parameters.

To derive the exact variance of the ATE estimator in a way that separates contributions from the two sampling stages, we introduce scaled version of the potential outcomes. For each unit $i$ in cluster $c$ and treatment status $d \in \{0,1\}$, define the scaled potential outcomes
\begin{align*}
   \tilY_{ci}(d) &= \frac{Y_{ci}(d)}{\Bar{N}}.
\end{align*}
The scaled cluster total is then
\begin{align*}
   \tilY_c(d) &= \sum_{i=1}^{N_c}\tilY_{ci}(d) = \frac{1}{\Bar{N}} \sum_{i=1}^{N_c} Y_{ci}(d),
\end{align*}
and the scaled cluster mean is
\begin{align*}
\bar{\tilY}_c(d) = \frac{1}{N_c}\tilY_c(d).
\end{align*}
The  mean of scaled cluster totals equals the mean outcome over all units.
\begin{align*}
\bar{{Y}}(d) = \frac{1}{C} \sum_{c=1}^C\tilY_c(d)=\frac{1}{N}\sum_{c=1}^C\sum_{i=1}^{N_c}Y_{ci}(d).
\end{align*}

With these definitions, the ATE $\tau$ can be expressed as an average of cluster-level treatment effects:
\begin{align*}
    \tau &= \frac{1}{C} \sum_{c=1}^C \left(\tilY_c(1) -\tilY_c(0) \right)=\frac{1}{C}\sum_{c=1}^C \tau_c,
\end{align*}
where $\tau_c =\tilY_c(1) -\tilY_c(0)$ is the treatment effect for cluster $c$.

We first consider an infeasible benchmark estimator. If all units within sampled clusters are observed, then the population  cluster totals $\tilY_{c}(1)$ and $\tilY_{c}(0)$ would be known. In that case, the infeasible Horvitz-Thompson estimator using the true cluster totals is
\begin{align}\label{tau-bar}
    \bar{\tau} &= \frac{1}{C} \sum_{c=1}^C \left[ \frac{R_c D_c\tilY_c}{p q} - \frac{R_c (1 - D_c)\tilY_c}{p (1 - q)} \right],
\end{align}
where $\tilY_c =D_c\tilY_c(1)+ (1-D_c)\tilY_c(0)$. In this expression, randomness arises only from the first-stage cluster sampling and treatment assignment of clusters, since no within-cluster sampling occurs.

In practice, the true cluster totals $\widetilde{Y}_c(d)$ are not observed because only a subset of units within each sampled cluster is surveyed. To obtain unbiased estimates of cluster-level totals, we use a within-cluster Horvitz-Thompson estimator:
\begin{align}\label{ht-sum}
    \widehat{\tilY}_c &= \sum_{i=1}^{N_c} \frac{R_{i \mid c}\tilY_{ci}}{\pi_c} = \frac{1}{\Bar{N}} \sum_{i=1}^{N_c} \frac{R_{i \mid c} Y_{ci}}{\pi_c}.
\end{align}
Replacing each unobserved cluster total $\tilY_c$ by its within-cluster estimate $\widehat{\tilY}_c$ yields an equivalent, cluster-level representation of the overall ATE estimator:
\begin{align}\label{HT-alt}
    \widehat{\tau} &= \frac{1}{C} \sum_{c=1}^C \left[ \frac{R_c D_c\widehat{\tilY}_c }{p q} - \frac{R_c (1 - D_c)\widehat{\tilY}_c} {p (1 - q)} \right].
\end{align}
This form makes explicit the two distinct sources of randomness in the estimator: between-cluster randomness, which is due to sampling and treatment assignment of clusters, and within-cluster randomness, which is due to the second-stage sampling of units used to construct $\widehat{\tilY}_c$. 

By expressing the estimator in this way, we can cleanly separate the contributions of each stage to the total design-based variance. The next subsection derives the exact variance of both $\bar{\tau}$ (the infeasible estimator) and $\widehat{\tau}$ (the feasible estimator with unit sampling).


\subsection{The exact variance}
We begin by showing that both $\bar\tau$ and $\widehat{\tau}$ are unbiased estimators of the average treatment effect, and then derive their exact variances. The results are summarized in the following proposition.

\begin{proposition}\label{prop:exact_var}
    Suppose the first stage sampling probability $p=\frac{S}{C}\in (0,1]$, and the cluster-level treatment assignment probability is $q = \frac{S_1}{S}\in (0,1)$. Then both estimators $\bar{\tau}$ and $\widehat{\tau}$ are unbiased for the finite-population average treatment effect $\tau$:
\begin{align*}
    \E\left[\bar{\tau}\right]=\tau, \hspace{0.1cm} \E\left[\widehat{\tau}\right] = \tau.
\end{align*}
    The exact variance of the infeasible estimator $\bar{\tau}$ (which uses true cluster totals) is 
    \begin{align}\label{eq:exact_var_1st}
        \Var\left(\bar{\tau}\right) &= \frac{1}{C}\left\{\frac{\frac{1}{C-1}\sum_{c=1}^C\left(\tilY_c(1)-\bar{{Y}}(1)\right)^2}{pq}+\frac{\frac{1}{C-1}\sum_{c=1}^C\left(\tilY_c(0)-\bar{{Y}}(0)\right)^2}{p(1-q)}-\frac{1}{C-1}\sum_{c=1}^C\left({\tau}_c-{\tau}\right)^2\right\}.
    \end{align}
When within-cluster unit sampling is also performed, let $\pi_c=\frac{n_c}{N_c}\in (0,1]$ denote the second-stage sampling probability and $\tilde{\pi}_c = \frac{n_c-1}{N_c-1}$. The exact variance of the feasible estimator $\widehat{\tau}$ is:
    \begin{align}\label{eq:exact_var_2nd}
        \Var\left(\widehat{\tau}\right) &= \frac{1}{C}\left\{\frac{\frac{1}{C-1}\sum_{c=1}^C\left(\tilY_c(1)-\bar{{Y}}(1)\right)^2}{pq}+\frac{\frac{1}{C-1}\sum_{c=1}^C\left(\tilY_c(0)-\bar{{Y}}(0)\right)^2}{p(1-q)}-\frac{1}{C-1}\sum_{c=1}^C\left({\tau}_c-{\tau}\right)^2\right\} \\ \nonumber
        & + \frac{1}{C^2}\left\{\sum_{c=1}^C\frac{(1-\pi_c)(1-\tilde{\pi}_c)}{\pi_c(\pi_c-\tilde{\pi}_c)}\left(\frac{\frac{1}{N_c-1}\sum_{i=1}^{N_c}(\tilY_{ci}(1)-\bar{\tilY}_c(1))^2}{pq} + \frac{\frac{1}{N_c-1}\sum_{i=1}^{N_c}(\tilY_{ci}(0)-\bar{\tilY}_c(0))^2}{p(1-q)}\right)\right\}.
    \end{align}
\end{proposition}

When only the first-stage sampling of clusters and the treatment assignment are random, the variance of $\bar{\tau}$ in (\ref{eq:exact_var_1st}) can be understood as having three additive components. The first term is between-cluster variation under treatment. $\frac{1}{C - 1} \sum_{c=1}^C \left(\tilY_c(1) - \bar{\tilY}(1) \right)^2$ measures the dispersion of the scaled cluster outcomes under treatment. It captures how heterogeneous treated clusters would be if all could be observed. The second term is between-cluster variation under control. $\frac{1}{C - 1} \sum_{c=1}^C \left(\tilY_c(0) - \bar{\tilY}(0) \right)^2$ is the analogous quantity for control clusters. The third term is the adjustment for the covariance between treatment and control potential outcomes. $\frac{1}{C - 1} \sum_{c=1}^C (\tau_c-\tau)^2$ equals $\frac{1}{C - 1} \sum_{c=1}^C \left(\tilY_c(1) - \bar{\tilY}(1) \right) \left(\tilY_c(0) - \bar{\tilY}(0) \right)$, and removes the double-counting of variability that arises when the two potential outcomes are positively correlated within clusters. Intuitively, when clusters that perform well under control also perform well under treatment, this term reduces the overall variance. The factors $1/(pq)$ and $1/(p(1-q))$ reflect the design probabilities: smaller sampling or treatment probabilities increase variance because fewer clusters contribute to information to each term. 

If there is no sampling, i.e., $S=C$ and all clusters are observed, our estimator $\bar{\tau}$ coincides with the difference-in-means estimator studied by \cite{su2021model}. The corresponding variance expression for $\bar{\tau}$, denoted as $\Var\left(\bar{\tau}\right)$, is identical to theirs, except that our expression includes the exact finite population correction factor of $\frac{1}{C-1}$ rather than $\frac{1}{C}$. \cite{abadie2023should} studied a difference-in-means estimator but under Bernoulli randomization, where each cluster is independently assigned to treatment or control. They approximate the variance by conditioning on the realized number of treated and control units. This yields a variance expression that is algebraically similar to ours, but omits the finite population correction, which aligns with the inherent nature of Bernoulli assignment. A similar distinction arises when only a subset of clusters is sampled in the first stage ($S<C$). In this case, the effective weights in the first two variance components of $\Var\left(\bar{\tau}\right)$ depend on both the cluster sampling and treatment probabilities-specifically, the products $pq$ and $p(1-q)$. These terms capture how the uncertainty from sampling and assignment jointly contributes to the overall design-based variance. 

When second-stage sampling is introduced, the true cluster totals $\tilY_c$ are no longer directly observed; instead, they are estimated from within-cluster samples using $\widehat{\tilY}_c$. Consequently, the variance of $\widehat{\tau}$ contains additional components that reflect the extra randomness introduced by the within-cluster sampling process ($\pi_c<1$). Intuitively, because each cluster's total outcome is now estimated from a subsample of its units, the precision of $\widehat{\tilY}_c$ depends on the within-cluster sampling fraction $\pi_c = n_c/N_c$, and the heterogeneity of outcomes among units in that cluster. These features appear in the additional variance terms involving the within-cluster variance of scaled outcomes, $\frac{1}{N_c - 1} \sum_{i=1}^{N_c} \left(\tilY_{ci}(d) - \bar{\tilY}_c(d) \right)^2$, which measure how dispersed individual outcomes are around the cluster mean. 

If all sampled clusters share the same within-cluster sampling i.e. $\pi_c=\pi$ for all $c$-then the adjustment term simplifies to a constant scaling factor. In that case, our variance expressions $\Var(\widehat{\tau})$ coincides with the one derived in \cite{abadie2023should}, except that our formulation explicitly includes finite population correction factors arising from considering sampling and treatment assignment without replacement.

\section{Asymptotic Analysis}\label{sec:asym}
We study the asymptotic properties of the Horvitz-Thompson estimator under a sequence of finite, clustered populations. The number of clusters $C$ grows to infinity, while the number of units per cluster may either remain fixed or grow with $C$. This is a finite population framework in which cluster counts, rather than individual counts, drive asymptotics. 

Our asymptotic analysis builds on two complementary strands of the literature. First, the survey-sampling strand, exemplified by \cite{ohlsson1989asymptotic} and \cite{chauvet2020inference}, establishes central limit theorems for two-stage sampling in finite populations. Second, \cite{li2017general} develops a general framework for finite-population asymptotics in design-based causal inference, identifying conditions under which estimators remain asymptotically normal as the population size grows. Related work in cluster-randomized experiments includes \cite{middleton2015unbiased} and \cite{su2021model}, who analyze asymptotics under cluster-level treatment assignment, and \cite{abadie2023should}, who develop a design‑based framework that explicitly incorporates both clustered sampling and clustered (possibly partial) assignment.

Building on these foundations, we establish asymptotic normality for our two‑stage Horvitz–Thompson estimator. To do so, we impose a set of regularity conditions on (i) moments of potential outcomes, (ii) sampling and treatment probabilities, and (iii) the degree of cluster‑size heterogeneity. These assumptions are stated next.
\begin{assumption}\label{ass:moment}
    $\frac{1}{N}\sum_{c=1}^C\sum_{i=1}^{N_c}Y_{ci}(d)^4=O(1),d=0,1$.
\end{assumption}
Assumption \ref{ass:moment} imposes a finite fourth-moment bound on unit-level potential outcomes. It ensures that extremely large or heavy-tailed outcomes do not dominate the behavior of the estimator. As in\cite{su2021model}, imposing the moment condition at the unit level ensures that the scaled cluster totals $\tilY_c(d)=\bar{N}^{-1}\sum_{i=1}^{N_c}Y_{ci}(d)$ are well-behaved. By contrast, \cite{abadie2023should} assume a stronger uniform boundedness condition on the potential outcomes—stronger than a finite fourth-moment requirement—thereby excluding heavy-tailed distributions. Our weaker condition allows for moderate heterogeneity across units and clusters, provided the overall moment bound holds.

We next state the assumptions on the sampling probabilities and treatment assignment.
\begin{assumption}\label{ass:cluster_sampling}
    The probability of cluster sampling $p=\frac{S}{C}\in(0,1]$ satisfies $p^{-1}=O\left(C^{\beta}\right)$, with $\beta\in [0,1/2)$.
\end{assumption}
\begin{assumption}\label{ass:treat_prob}
The probability of treatment $q=\frac{S_1}{S}\in(0,1)$.
\end{assumption}
\begin{assumption}\label{ass:sample_units}
The second-stage sampling probability is not $0$: $n_c>0$, $\forall c$.
\end{assumption}
Assumption \ref{ass:cluster_sampling} allows the first-stage sampling rate $p$ to decline with $C$, but not too quickly. The restriction $p^{-1}=O(C^\beta)$ with $\beta<1/2$ ensures that the number of sampled clusters $S=pC$ still grows sufficiently fast to support a central limit theorem. Assumption \ref{ass:treat_prob} is the standard overlap condition, ruling out near‑degenerate treatment shares among sampled clusters. Assumption \ref{ass:sample_units} requires the number of units sampled in the sampled clusters to be bounded away from $0$. This is a minimal requirement preventing empty clusters in the observed data. In practice, it is usually automatically satisfied since the within-cluster sample sizes $n_c$ are fixed in advance at the design stage.

We also need to control cluster-size heterogeneity, which plays a critical role in establishing asymptotic normality. Let $\omega_c = N_c/\Bar{N}$, where $\Bar{N}$ is the average cluster size, and define $\omega = \max_c \omega_c$. Here, $\omega_c$ is cluster $c$'s relative size compared to the average, while $\omega$ captures the overall degree of heterogeneity in cluster size. If all clusters are of equal size, i.e., $N_c=\bar{N}$, then $\omega=1$. Large values of $\omega$ indicate that a few clusters are disproportionately large, which can dominate the sampling variation and invalidate standard central limit approximations if left uncontrolled.

\begin{assumption}\label{ass:size}
    $\omega=o\left(C^{\frac{1}{3}\left(1-2\beta\right)}\right)$.
\end{assumption}

In comparison, \cite{abadie2023should} assumed that $\frac{\max_c N_c}{\min_c N_c}<\infty$, which implies that cluster sizes are uniformly bounded relative to each other. Our condition is weaker: we permit some heterogeneity and even allow the largest cluster size to grow with $C$, provided it does so slowly enough. When all clusters are sampled ($p=1$), our rate condition simplifies to $\omega=o\left(C^{1/3}\right)$, which coincides with the assumption in \cite{su2021model}. Because our setting includes cluster sampling, we must simultaneously control the sampling probability $p$ (Assumption \ref{ass:cluster_sampling}) and the degree of size heterogeneity $\omega$ (Assumption \ref{ass:size}). Because our setting includes cluster sampling, we must simultaneously control the sampling rate: $p^{-1}\omega^{3/2}=o\left(C^{1/2}\right)$. Intuitively, the condition means that as the number of clusters grows, both the sampling rate and the distribution of cluster sizes must be balanced so that the total information is not concentrated in a small number of large clusters.

\begin{lemma}\label{lm:clt_s1}
     Under Assumptions \ref{ass:moment}, \ref{ass:cluster_sampling}, \ref{ass:treat_prob}, and \ref{ass:size}, if $C\Var(\bar{\tau})\nrightarrow 0$ as $C\rightarrow\infty$, then $(\bar{\tau}-\tau)/\text{se}(\bar{\tau})\rightarrow^d \mathcal{N}(0,1)$, where $\text{se}(\bar{\tau})=\sqrt{\Var(\bar{\tau})}$.
\end{lemma}
The condition $C\Var(\bar{\tau})\nrightarrow 0$ ensures a non-degenerate design variance, which means that the variance arising from the first-stage sampling and treatment assignment does not vanish as the number of clusters grows. This assumption is the finite-population analogue of requiring that the design provides enough independent information for a central limit theorem to apply, much as in classical two-stage survey-sampling theory.

Define the across-cluster variance and covariance components:
\begin{align*}
    &\sigma^2(1)=\frac{1}{C-1}\sum_{c=1}^C\left(\tilY_c(1)-\bar{Y}(1)\right)^2, \quad \sigma^2(0)=\frac{1}{C-1}\sum_{c=1}^C\left(\tilY_c(0)-\bar{Y}(0)\right)^2\\
    &\sigma^2(\tau)=\frac{1}{C-1}\sum_{c=1}^C\left(\tau_c-\tau\right)^2, \quad \sigma_{10} =\frac{1}{C-1}\sum_{c=1}^C\left(\tilY_c(1)-\bar{Y}(1)\right)\left(\tilY_c(0)-\bar{Y}(0)\right).
\end{align*}

A sufficient condition for $C\Var(\bar{\tau})\nrightarrow 0$ is either: (i) $\sigma^2(1)\nrightarrow 0$ and $(1-p)+\left|\lambda-\frac{1-q}{q}\right|\neq 0$, where $\lambda = \frac{\sigma(0)}{\sigma(1)}$; or (ii) $\sigma^2(0)\nrightarrow 0$ and $(1-p)+\left|\lambda-\frac{q}{1-q}\right|\neq 0$, where $\lambda = \frac{\sigma(1)}{\sigma(0)}$. These conditions ensure that the design variance does not collapse wither because treatment shares become perfectly balanced ($q=1/2$) or because sampling covers all clusters ($p=1$) while the treated and control potential outcomes have co-movements. 

To see why the above suffices, suppose $\sigma^2(1)\nrightarrow 0$ and define $\lambda = \sigma(0)/\sigma(1)$. Using $\sigma^2(\tau) = \sigma^2(1)+\sigma^2(0)-2\sigma_{10}$, we could have
\begin{align*}
    C\Var(\bar{\tau}) &= \frac{1}{pq}\sigma^2(1) + \frac{1}{p(1-q)}\sigma^2(0) - \sigma^2(\tau) \\
    &= \frac{1}{pq}\sigma^2(1) + \frac{1}{p(1-q)}\sigma^2(0) - \left(\sigma^2(1)+\sigma^2(0)-2\sigma_{10}\right) \\
    &\geq \left(\frac{1}{pq}-1\right)\sigma^2(1) + \left(\frac{1}{p(1-q)}-1\right)\lambda^2\sigma^2(1) - 2\lambda\sigma^2(1) \\
    &= \left[\left(\frac{1}{p(1-q)}-1\right)\lambda^2 - 2\lambda + \left(\frac{1}{pq}-1\right)\right]\sigma^2(1)
\end{align*}
where the inequality uses $\sigma_{10}\geq -\sigma(1)\sigma(0)=-\lambda \sigma^2(1)$ by Cauchy-Schwarz. The expression in the bracket, $\left(\frac{1}{p(1-q)}-1\right)\lambda^2 - 2\lambda + \left(\frac{1}{pq}-1\right)$, is equal to $\left(\sqrt{\frac{q}{1-q}}\lambda-\sqrt{\frac{1-q}{q}}\right)^2$ if $p=1$; and it is strictly positive if $p<1$.

The proof of Lemma \ref{lm:clt_s1} follows by applying the finite-population CLT in \cite{li2017general} to the cluster-level Hotvitz-Thompson statistics. The argument parallels Theorem 1 of \cite{su2021model}, who derive the same result for cluster-randomized experiment without a cluster sampling stage. 

We now turn to the feasible estimator $\widehat{\tau}$. Relative to $\Bar{\tau}$, the estimator $\widehat{\tau}$ includes second-stage unit sampling, which introduces additional randomness and within-cluster dependence. Proving asymptotic normality therefore requires two ingredients: (i) the CLT for the first‑stage component $\Bar{\tau}$ (Lemma \ref{lm:clt_s1}); (ii) a Lyapunov-type condition ensuring that the added variability from estimating cluster totals using within-cluster subsamples is asymptotically negligible—formally analogous to conditions in two-stage survey-sampling CLTs. Together, these will yield asymptotic normality of $\widehat{\tau}$ as $C\rightarrow\infty$.

\begin{theorem}\label{thm:clt_s2}
     Under Assumption \ref{ass:moment}, \ref{ass:cluster_sampling}, \ref{ass:treat_prob}, \ref{ass:sample_units}, and \ref{ass:size}, and conditions specified in Lemma \ref{lm:clt_s1}, if $C\Var(\bar{\tau})\nrightarrow 0$ as $C\rightarrow \infty$, then $(\widehat{\tau} -\tau)/ \text{se}(\widehat{\tau})\rightarrow^d \mathcal{N}(0,1)$, where $\text{se}(\widehat{\tau})=\sqrt{\Var(\widehat{\tau})}$.
\end{theorem}
Theorem \ref{thm:clt_s2} generalizes classical two‑stage survey‑sampling CLTs \citep{ohlsson1989asymptotic,chauvet2020inference} to a causal estimand with cluster‑level assignment. It also parallels the main asymptotic normality result of
\citet{abadie2023should}, who consider clustered sampling and assignment in a more general design-based framework.
Our analysis differs, however, in focusing specifically on cluster-randomized experiments with explicit two-stage sampling. Consequently, our rate conditions are expressed in terms of the number of clusters and the heterogeneity of their sizes, rather than only the total number of sampled units. \citet{abadie2023should} likewise require the number of sampled clusters to diverge, but because their framework is broader (not limited to CRTs), their asymptotic statements are framed in a more general sampling and treatment assignment environment.

Lemma \ref{lm:clt_s1} and Theorem \ref{thm:clt_s2} jointly show that under mild design-based conditions, the Horvitz–Thompson estimator is asymptotically normal—first when cluster totals are known, and then when they are estimated via within-cluster sampling. In empirical terms, this result justifies using standard normal-based inference for two-stage cluster-randomized designs, provided that the number of sampled clusters is sufficiently large and no small subset of clusters dominates the sample.

\section{Estimators of Exact Variance}\label{sec:bound}
In clustered randomized controlled trials, we never observe both $\tilY_c(1)$ or $\tilY_c(0)$ for the same cluster $c$. Even without within-cluster subsampling, each cluster reveals only one potential cluster total. Following \citet{aronow2014sharp}, we construct sharp bounds on the exact variance of $\widehat{\tau}$. 

Throughout, expectations $\mathbb{E}[\cdot]$ and covariances are design-based (i.e., taken over the sampling and assignment mechanisms). For brevity, define the design covariances
\[
\Delta_{cc'}^1 = \mathbb{E}[ R_c D_c R_{c'} D_{c'} ] - \mathbb{E}[ R_c D_c ] \mathbb{E}[ R_{c'} D_{c'} ],
\]
\[
\Delta_{cc'}^0 = \mathbb{E}[ R_c (1 - D_c) R_{c'} (1 - D_{c'}) ] - \mathbb{E}[ R_c (1 - D_c) ] \mathbb{E}[ R_{c'} (1 - D_{c'}) ],
\]
and, for the second stage,
\[
\Delta_{ij \mid c} = \mathbb{E}[ R_{i \mid c} R_{j \mid c} ] - \mathbb{E}[ R_{i \mid c} ] \mathbb{E}[ R_{j \mid c} ].
\]

A convenient rearrangement of the exact variance of $\widehat{\tau}$ is
\begin{align*}
    \Var\left( \widehat{\tau} \right)& = V_1 + V_0 + \frac{2}{C}\sigma\left(\tilY(1),\tilY(0)\right),
\end{align*}
where
\begin{align*}
    V_1 &= \frac{1}{C^2} \Bigg\{ \sum_{c=1}^C \sum_{c'=1}^C \frac{\Delta_{cc'}^1  {\tilY}_c(1) {\tilY}_{c'}(1)}{\mathbb{E}[ R_c D_c ] \mathbb{E}[ R_{c'} D_{c'} ] }  + \sum_{c=1}^C \sum_{i=1}^{N_c} \sum_{j=1}^{N_c} \frac{ \Delta_{ij \mid c} \tilY_{ci}(1)\tilY_{cj}(1) }{ \mathbb{E}[ R_{i \mid c} ] \mathbb{E}[ R_{j \mid c} ] }\Bigg\}, \\
    V_0 &= \frac{1}{C^2} \Bigg\{ \sum_{c=1}^C \sum_{c'=1}^C \frac{\Delta_{cc'}^0 {\tilY}_c(0) {\tilY}_{c'}(0)}{\mathbb{E}[ R_c (1-D_c) ] \mathbb{E}[ R_{c'} (1-D_{c'}) ] }  + \sum_{c=1}^C \sum_{i=1}^{N_c} \sum_{j=1}^{N_c} \frac{ \Delta_{ij \mid c} \tilY_{ci}(0)\tilY_{cj}(0) }{ \mathbb{E}[ R_{i \mid c} ] \mathbb{E}[ R_{j \mid c} ] }  \Bigg\},
\end{align*}
and
\begin{align*}
    \sigma\left(\tilY(1),\tilY(0)\right) &= \frac{1}{C-1}\left(\sum_{c=1}^C{\tilY}_c(1){\tilY}_{c}(0) - \left(\frac{1}{C}\sum_{c=1}^C{\tilY}_c(1)\right)\left(\frac{1}{C}\sum_{c=1}^C{\tilY}_{c}(0)\right)\right).
\end{align*}
Here, $V_1$ and $V_0$ collect the first-stage (cluster) contributions together with the second-stage (unit) contributions within each arm. The term $\sigma\left(\tilY(1),\tilY(0)\right)$ denotes the cross-arm covariance of cluster totals, taken across clusters. 

With unit-level sampling within the sampled clusters, $\tilY_c(1)$ or $\tilY_c(0)$ cannot be observed for any $c$ unless the second-stage sampling probability $\pi_c=1$. Instead, we only observe  the within-cluster Horvitz-Thompson estimates of the scaled cluster totals, denoted by $\widehat{\tilY}_c(d)$. If, counterfactually, we could draw independent subsamples of units from each cluster to estimate both $\widehat{\tilY}_c(1)$ and $\widehat{\tilY}_c(0)$ for the same cluster $c$, then $\E\left[\widehat{\tilY}_c(1)\widehat{\tilY}_{c'}(0)\right] = \tilY_c(1)\tilY_{c'}(0)$, $\forall c, c'$. In that hypothetical setting, a plug-in estimator 
\begin{align*}
    \widetilde{\Var}\left( \widehat{\tau} \right) 
    &= \widetilde{V}_1 + \widetilde{V}_0 + \frac{2}{C}\sigma\left(\widehat{\tilY}(1),\widehat{\tilY}(0)\right),
\end{align*}
would be unbiased and consistent for ${\Var}\left( \widehat{\tau} \right)$. Here,
\begin{align*}
&\widetilde{V}_1 = \frac{1}{C^2} \Bigg\{ \sum_{c=1}^C \sum_{c'=1}^C \frac{\Delta_{cc'}^1  \widehat{\tilY}_c(1) \widehat{\tilY}_{c'}(1)}{\mathbb{E}[ R_c D_c ] \mathbb{E}[ R_{c'} D_{c'} ] }  + \sum_{c=1}^C \sum_{i=1}^{N_c} \sum_{j=1}^{N_c} \frac{ \Delta_{ij \mid c} R_{i \mid c}R_{j \mid c}\tilY_{ci}(1)\tilY_{cj}(1) }{ \mathbb{E}[ R_{i \mid c}R_{j \mid c}  ]\mathbb{E}[ R_{i \mid c} ] \mathbb{E}[ R_{j \mid c} ] }  \Bigg\}, \\
& \widetilde{V}_0  = \frac{1}{C^2} \Bigg\{ \sum_{c=1}^C \sum_{c'=1}^C \frac{\Delta_{cc'}^0 \widehat{\tilY}_c(0) \widehat{\tilY}_{c'}(0)}{\mathbb{E}[ R_c (1-D_c) ] \mathbb{E}[ R_{c'} (1-D_{c'}) ] }  + \sum_{c=1}^C \sum_{i=1}^{N_c} \sum_{j=1}^{N_c} \frac{ \Delta_{ij \mid c} R_{i \mid c}R_{j \mid c}\tilY_{ci}(0)\tilY_{cj}(0) }{ \mathbb{E}[ R_{i \mid c}R_{j \mid c}  ]\mathbb{E}[ R_{i \mid c} ] \mathbb{E}[ R_{j \mid c} ] }  \Bigg\}, \\
& \sigma\left(\widehat{\tilY}(1),\widehat{\tilY}(0)\right) = \frac{1}{C-1}\left(\sum_{c=1}^C\widehat{\tilY}_c(1)\widehat{\tilY}_{c}(0) - \left(\frac{1}{C}\sum_{c=1}^C\widehat{\tilY}_c(1)\right)\left(\frac{1}{C}\sum_{c=1}^C\widehat{\tilY}_{c}(0)\right)\right).
\end{align*}

In practice, however, we cannot observe both potential cluster totals for any given cluster. With second-stage sampling, each cluster provides only one observed aggregate, either $\widehat{\widetilde{Y}}_c(1)$ or $\widehat{\widetilde{Y}}_c(0)$, depending on its assignment. Formally,
\begin{align*}
\widehat{\widetilde{Y}}_c(d)=\sum_{i=1}^{N_c}\frac{R_{i|c}\widetilde{Y}_{ci}(d)}{\pi_c} = \frac{1}{\Bar{N}}sum_{i=1}^{N_c}\frac{R_{i|c}\widetilde{Y}_{ci}(d)}{\pi_c}, d\in\{0,1\}.
\end{align*}

Hence, $\widetilde{\Var}\left( \widehat{\tau} \right)$ is infeasible because it relied on cross-arm covariance terms involving both $\widehat{\widetilde{Y}}_c(1)$ and $\widehat{\widetilde{Y}}_c(0)$ for the same $c$. To address this, we apply Fr\'echet-Hoeffding covariance bounds proposed by \cite{aronow2014sharp} to construct an interval estimator for this infeasible variance. These bounds describe the maximal and minimal possible covariance consistent with the observed marginal distributions of $\widehat{\widetilde{Y}}_c(1)$ and $\widehat{\widetilde{Y}}_c(0)$.

We next characterize the maximal and minimal possible covariance between the estimated cluster aggregates under treatment and control. Define the empirical cumulative distribution functions (CDFs) of $\widehat{\tilY}_c(1)$ and $\widehat{\tilY}_c(0)$ as
\[
G(y) = \frac{1}{C} \sum_{c=1}^C \mathbf{1}\left(\widehat{\tilY}_c(1) \leq y \right), \quad F(y) = \frac{1}{C} \sum_{c=1}^C \mathbf{1}\left(\widehat{\tilY}_c(0) \leq y \right),
\]
where $\mathbf{1}(\cdot)$ is the indicator function. The corresponding left-continuous inverses (quantile functions) of these CDFs are defined as
\[
G^{-1}(u) = \inf \{ y : G(y) \geq u \}, \quad F^{-1}(u) = \inf \{ y : F(y) \geq u \}, \quad \text{for } u \in [0, 1].
\]

Using these marginal distributions, the Fr\'echet-Hoeffding bounds identify the extreme possible covariances between $\widehat{\tilY}_c(1)$ and $\widehat{\tilY}_c(0)$ compatible with the observed marginals: We then define the maximal and minimal possible covariances:

\textit{Maximal Covariance (Positive Association):}
\[
\sigma^H\left(\widehat{\tilY}(1),\widehat{\tilY}(0) \right) = \int_0^1 G^{-1}(u) F^{-1}(u) \, du - \left(\frac{1}{C}\sum_{c=1}^C\widehat{\tilY}_c(1)\right)\left(\frac{1}{C}\sum_{c=1}^C\widehat{\tilY}_{c}(0)\right).
\]

\textit{Minimal Covariance (Negative Association):}
\[
\sigma^L\left(\widehat{\tilY}(1),\widehat{\tilY}(0) \right) = \int_0^1 G^{-1}(u) F^{-1}(1 - u) \, du -  \left(\frac{1}{C}\sum_{c=1}^C\widehat{\tilY}_c(1)\right)\left(\frac{1}{C}\sum_{c=1}^C\widehat{\tilY}_{c}(0)\right).
\]

These bounds provide the most precise interval for $\sigma\left(\widehat{\tilY}(1),\widehat{\tilY}(0) \right)$ given the observed data, without making additional assumptions about the joint distribution of $\widehat{\tilY}(1)$ and $\widehat{\tilY}(0)$.  Different from \cite{aronow2014sharp}, we are now constructing bounds on the covariance between $\widehat{\tilY}_c(1)$ and $\widehat{\tilY}_c(0)$, instead of the covariance between true scaled cluster totals. 

By replacing $\sigma\left(\widehat{\tilY}(1),\widehat{\tilY}(0) \right)$ in $\widetilde{\Var}\left(\widehat{\tau}\right)$ with $\sigma^H\left(\widehat{\tilY}(1),\widehat{\tilY}(0) \right)$ and $\sigma^L\left(\widehat{\tilY}(1),\widehat{\tilY}(0) \right)$ respectively, we could obtain $\widetilde{\Var}^H\left(\widehat{\tau}\right)$ and $\widetilde{\Var}^L\left(\widehat{\tau}\right)$, delivering the sharpest bounds on $\widetilde{\Var}\left(\widehat{\tau}\right)$.

In practice, we do not observe both potential aggregates for any cluster, since we only observe $\widehat{\tilY}_c(1)$ for treated sampled clusters and $\widehat{\tilY}_c(0)$ for control sampled clusters. We therefore construct the empirical cumulative distribution functions of the estimated totals $\widehat{\tilY}_c(1)$ and $\widehat{\tilY}_c(0)$ as 
\[
\widehat{G}(y) = \frac{1}{C} \sum_{c=1}^C \frac{R_c D_c}{p q} \mathbf{1} \left\{ \widehat{\tilY}_c(1) \leq y \right\}, \quad
\widehat{F}(y) = \frac{1}{C} \sum_{c=1}^C \frac{R_c (1 - D_c)}{p (1 - q)} \mathbf{1} \left\{ \widehat{\tilY}_c(0) \leq y \right\}.
\]

These weights ensure that the empirical distributions of the observed treated and control clusters estimate the population marginals of $\widehat{\tilY}_c(1)$ and $\widehat{\tilY}_c(0)$. The corresponding left-continuous inverses of $\widehat{G}$ and $\widehat{F}$ are
\[
\widehat{G}^{-1}(u) = \inf \left\{ y : \widehat{G}(y) \geq u \right\} = \widehat{\tilY}_{(\lceil S_1 u \rceil)}(1),
\quad
\widehat{F}^{-1}(u) = \inf \left\{ y : \widehat{F}(y) \geq u \right\} = \widehat{\tilY}_{(\lceil S_0 u \rceil)}(0),
\]
where $\widehat{\tilY}_{(1)}(1) \leq \dots \leq \widehat{\tilY}_{(S_1)}(1)$ and $\widehat{\tilY}_{(1)}(0) \leq \dots \leq \widehat{\tilY}_{(S_0)}(0)$ are the ordered aggregates for the $S_1$ treated and $S_0$ control clusters. 

To approximate the integrals above, let $\mathcal{B} = \{ b_0, b_1, \dots, b_B \}$ be the ordered distinct values in the set $\{ 0, 1/S_1, 2/S_1, \dots, 1 \} \cup \{ 0, 1/S_0, 2/S_0, \dots, 1 \}$, which partitions $[0, 1]$. 

Define the corresponding stepwise quantiles
\[
\widehat{\tilY}_{[h]}(1) = \widehat{\tilY}_{(\lceil S_1 b_h \rceil)}(1), \quad \widehat{\tilY}_{[h]}(0) = \widehat{\tilY}_{(\lceil S_0 b_h \rceil)}(0), \quad h = 1, \dots, B.
\]
Then the sample analogues of the covariance bounds are 
\begin{align*}
\widehat{\sigma}^{H}\left( \widehat{\tilY}(1), \widehat{\tilY}(0) \right) &= \sum_{h=1}^B (b_h - b_{h-1}) \widehat{\tilY}_{[h]}(1) \widehat{\tilY}_{[h]}(0) - \left( \frac{1}{C} \sum_{c=1}^C \frac{R_c D_c \widehat{\tilY}_c }{pq} \right) \left( \frac{1}{C} \sum_{c=1}^C \frac{R_c (1 - D_c)\widehat{\tilY}_c}{p(1-q)}  \right), \\
\widehat{\sigma}^{L}\left( \widehat{\tilY}(1), \widehat{\tilY}(0) \right) &= \sum_{h=1}^B (b_h - b_{h-1}) \widehat{\tilY}_{[h]}(1) \widehat{\tilY}_{[B + 1 - h]}(0) - \left( \frac{1}{C} \sum_{c=1}^C \frac{R_c D_c \widehat{\tilY}_c }{pq} \right) \left( \frac{1}{C} \sum_{c=1}^C \frac{R_c (1 - D_c)\widehat{\tilY}_c}{p(1-q)}  \right).
\end{align*}

Again, with the first-stage sampling and the treatment assignment, $\widetilde{V}_1$ and $\widetilde{V}_0$ is not directly observed. An unbiased and consistent estimator for $\widetilde{V}_1$ is 
\begin{align*}
    \widehat{V}_1 = \frac{1}{C^2} & \Bigg\{ \sum_{c=1}^C \sum_{c'=1}^C \frac{\Delta_{cc'}^1 R_c D_c R_{c'} D_{c'} \widehat{\tilY}_c \widehat{\tilY}_{c'}}{ \mathbb{E}[ R_c D_c R_{c'} D_{c'} ] \mathbb{E}[ R_c D_c ] \mathbb{E}[ R_{c'} D_{c'} ] }  \\
     &+ \sum_{c=1}^C \frac{ R_c D_c }{ \mathbb{E}[ R_c D_c ] }\sum_{i=1}^{N_c} \sum_{j=1}^{N_c} \frac{ \Delta_{ij \mid c} R_{i \mid c} R_{j \mid c}\tilY_{ci}\tilY_{cj} }{ \mathbb{E}[ R_{i \mid c} R_{j \mid c} ] \mathbb{E}[ R_{i \mid c} ] \mathbb{E}[ R_{j \mid c} ] } \Bigg\}.
\end{align*}
Similarly, an unbiased and consistent estimator for $\widetilde{V}_0$ is 
\begin{align*}
    \widehat{V}_0 = \frac{1}{C^2} &\Bigg\{ \sum_{c=1}^C \sum_{c'=1}^C \frac{\Delta_{cc'}^0 R_c (1-D_c) R_{c'} (1-D_{c'}) \widehat{\tilY}_c \widehat{\tilY}_{c'}}{ \mathbb{E}[ R_c (1-D_c) R_{c'} (1-D_{c'}) ] \mathbb{E}[ R_c (1-D_c) ] \mathbb{E}[ R_{c'} (1-D_{c'}) ] } \\
    & + \sum_{c=1}^C \frac{ R_c (1-D_c) }{ \mathbb{E}[ R_c (1-D_c)  ] }\sum_{i=1}^{N_c} \sum_{j=1}^{N_c} \frac{ \Delta_{ij \mid c} R_{i \mid c} R_{j \mid c}\tilY_{ci}\tilY_{cj} }{ \mathbb{E}[ R_{i \mid c} R_{j \mid c} ] \mathbb{E}[ R_{i \mid c} ] \mathbb{E}[ R_{j \mid c} ] } \Bigg\}.
\end{align*}
The above variance estimators are the so-called Horvitz-Thompson variance estimator, and the counterparts in survey sampling can be found in \cite{chauvet2020inference}. Each line mirrors the classic HT variance form: a cluster-level piece that uses pairwise cluster inclusion covariances $\Delta_{cc'}^d$, and a unit-level piece that uses pairwise unit inclusion covariances $\Delta_{ij|c}$, with the appropriate design probabilities in the denominators. Under our regular conditions, these estimators are unbiased for the corresponding variance components and consistent as $C\rightarrow\infty$.

Combining the arm-specific pieces with the sharp covariance bound estimators $\widehat{\sigma}^{H}$ and $\widehat{\sigma}^{L}$,  yields an interval estimator ${\widehat{\Var}^{H}(\widehat{\tau})}$ and ${\widehat{\Var}^{L}(\widehat{\tau})}$ as:
\begin{align*}
{\widehat{\Var}^{H}(\widehat{\tau})} &= \widehat{V}_1 + \widehat{V}_0 + \frac{2}{C}\widehat{\sigma}^{H}\left( \widehat{\tilY}(1), \widehat{\tilY}(0) \right) \\
{\widehat{\Var}^{H}(\widehat{\tau})} &= \widehat{V}_1 + \widehat{V}_0 + \frac{2}{C}\widehat{\sigma}^{L}\left( \widehat{\tilY}(1), \widehat{\tilY}(0) \right).
\end{align*}

\begin{proposition}\label{prop:bound}
    Under Assumption \ref{ass:moment} - \ref{ass:size}, $C{\widehat{\Var}^{H}(\widehat{\tau})} \xrightarrow{p} C{\widetilde{\Var}^{H}(\widehat{\tau})}$ and $C\widehat{\Var}^L(\widehat{\tau})\xrightarrow{p} C\widetilde{\Var}^L(\widehat{\tau})$ as $C \to \infty$; that is, $\left[{\widehat{\Var}^{L}(\widehat{\tau})}-{\widetilde{\Var}^{L}(\widehat{\tau})}, {\widehat{\Var}^{H}(\widehat{\tau})}-{\widetilde{\Var}^{H}(\widehat{\tau})}\right]=o_p(1)$.
\end{proposition}

\begin{remark}
    An alternative route is to apply the \cite{aronow2014sharp} bounds directly to the covariance of the true (unobserved) cluster totals. This yields variance bounds $\left[\Var^L\left(\widehat{\tau}\right),\Var^H\left(\widehat{\tau}\right)\right]$. To estimate those bounds in our setting, two adjustments are required: (i) replace the unobserved marginals of those of $\widetilde{Y}_c(d)$ by the marginals of the estimated cluster totals $\widehat{\widetilde{Y}}_c(d)$, which incorporate second-stage within-cluster sampling; and (ii) move from the population of clusters to the sampled clusters using HT weights, which account for first-stage sampling and treatment assignment. Operationally, these steps lead to the same feasible interval estimators we propose. Under stronger growth conditions, i.e., $C\rightarrow\infty$ and $\min_cN_c\rightarrow \infty$ (so that within-cluster estimation error vanishes and the empirical marginals of $\widehat{\widetilde{Y}}_c(d)$ converge to those of ${\widetilde{Y}}_c(d)$, the interval estimators converge to $\left[\Var^L\left(\widehat{\tau}\right),\Var^H\left(\widehat{\tau}\right)\right]$. 
    
    Our approach, however, deliberately targets the bounds $\left[\widetilde{\Var}^L\left(\widehat{\tau}\right),\widetilde{\Var}^H\left(\widehat{\tau}\right)\right]$ defined with the estimated cluster aggregates. This choice has a practical advantage: it yields consistent bound estimators without requiring $\min_c N_c\rightarrow \infty$; that is, cluster sizes may remain bounded while the number of clusters grow.
\end{remark}

Additionally, dropping the negative term $-\frac{1}{C-1}\sum_{c=1}^C\left(\tau_c-\tau\right)^2$ in (\ref{eq:exact_var_2nd}), we could obtain a conservative part of the exact variance, denoted as $\Var(\widehat{\tau})^{\text{consv}}$:
\begin{align*}
\Var\left( \widehat{\tau} \right)^{\text{consv}} &= \frac{1}{C}\frac{\frac{1}{C-1}\sum_{c=1}^C\left(\tilY_c(1)-\bar{{Y}}(1)\right)^2}{pq} + \frac{1}{C^2}\sum_{c=1}^C\frac{(1-\pi_c)(1-\tilde{\pi}_c)}{\pi_c(\pi_c-\tilde{\pi}_c)}\frac{\frac{1}{N_c-1}\sum_{i=1}^{N_c}(\tilY_{ci}(1)-\bar{\tilY}_c(1))^2}{pq} \\
&+\frac{1}{C}\frac{\frac{1}{C-1}\sum_{c=1}^C\left(\tilY_c(0)-\bar{{Y}}(0)\right)^2}{p(1-q)} +  \frac{1}{C^2}\sum_{c=1}^C\frac{(1-\pi_c)(1-\tilde{\pi}_c)}{\pi_c(\pi_c-\tilde{\pi}_c)} \frac{\frac{1}{N_c-1}\sum_{i=1}^{N_c}(\tilY_{ci}(0)-\bar{\tilY}_c(0))^2}{p(1-q)}.
\end{align*}
If treatment effects are homogeneous across clusters, then $\Var\left( \widehat{\tau} \right)^{\text{consv}} = \Var\left( \widehat{\tau} \right)$, this conservative variance part is exact. An estimator for the above conservative part of the exact variance is
\begin{align*}
\widehat{\Var}(\widehat{\tau})^{\text{consv}}& = \frac{1}{C^2} \Bigg\{  \frac{1}{1 -pq} \sum_{c=1}^C \sum_{c'=1}^C \frac{ \Delta_{cc'}^1 R_c D_c R_{c'} D_{c'} \widehat{\tilY}_c \widehat{\tilY}_{c'} }{ \mathbb{E}[ R_c D_c R_{c'} D_{c'} ] \mathbb{E}[ R_c D_c ] \mathbb{E}[ R_{c'} D_{c'} ] } \\
& + \frac{1}{1 - p(1-q)} \sum_{c=1}^C \sum_{c'=1}^C \frac{ \Delta_{cc'}^0 R_c (1 - D_c) R_{c'} (1 - D_{c'}) \widehat{\tilY}_c \widehat{\tilY}_{c'} }{ \mathbb{E}[ R_c (1 - D_c) R_{c'} (1 - D_{c'}) ] \mathbb{E}[ R_c (1 - D_c) ] \mathbb{E}[ R_{c'} (1 - D_{c'}) ] }.
\end{align*}

\begin{proposition}\label{prop:conserve}
$C\Var(\widehat{\tau})\leq C\Var(\widehat{\tau})^{\text{consv}}$. If the fourth moment of the potential outcomes is bounded, and $\left(\frac{C}{S_d}-1\right)\sigma^2\left(\tilY(d)\right)\nrightarrow 0$, $d=0,1$, then $C\widehat{\Var}(\widehat{\tau})^{\text{consv}}$ is an unbiased and consistent estimator of $C\Var(\widehat{\tau})^{\text{consv}}$.
\end{proposition}

\section{Simulation Results} \label{sec:simulation}
This section investigates the performance of our proposed bound estimators through Monte Carlo simulations. We compare it against the widely used Liang-Zeger cluster-robust standard error estimator. 

For each unit $i$ in cluster $c$, we generate two independent covariates $X_{ci}=(X_{1ci},X_{2ci})\in (U[0,1])^2$ and set $\gamma=(1,1)$. We consider four data generating processes (DGPs) that differ in the extent and source of treatment effect heterogeneity and clustering. All idiosyncratic noise terms are Gaussian with variance $25$. The DGPs are summarized in Table~\ref{tab:dgp}.
\begin{table}[htbp]
\renewcommand{\arraystretch}{1.4}
\centering
\begin{tabular}{c|c|c}
\hline
\hline
     & $Y_{ci}(0)$ & $Y_{ci}(1)$ \\
     \hline 
    DGP 1 & $\mathcal{N}(\gamma^TX_{ci},25)$ & $\mathcal{N}(\tau+\gamma^TX_{ci},25)$ \\
    DGP 2 & $\mathcal{N}(\gamma^TX_{ci},25)$ & $\mathcal{N}(\tau_{i}+\gamma^TX_{ci},25)$ \\
    DGP 3 & $\mathcal{N}(-\alpha_c+\gamma^TX_{ci},25)$ & $\mathcal{N}(\alpha_c+\tau_c+\gamma^TX_{ci},25)$\\
    DGP 4 & $\mathcal{N}(-\alpha_c+\gamma^TX_{ci},25)$ & $\mathcal{N}(\alpha_c+\tau_{ci}+\gamma^TX_{ci},25)$\\
    \hline
\end{tabular}
\caption{Data Generating Processes}
\caption*{\textit{DGP~1 is a benchmark with a constant treatment effect $\tau=50$. DGP~2 introduces individual-level heterogeneity without clustering, with $\tau_i\sim\mathcal{N}(50,100)$. DGP~3 incorporates cluster-level heterogeneity through a cluster-specific intercept $\alpha_c\sim\mathcal{N}(5,25)$ and a cluster-level treatment effect $\tau_c\sim\mathcal{N}(20,100)$; units within the same cluster share the same treatment effect, but effects vary across clusters. DGP~4 allows for both between- and within-cluster heterogeneity: cluster means vary as $\tau_c\sim\mathcal{N}(20,100)$, and individual treatment effects vary within clusters as $\tau_{ci}\sim\mathcal{N}(\tau_c,\sigma_c^2)$ with $\sigma_c^2\sim U[0,4]$. }}
\label{tab:dgp}
\end{table}
The four DGPs span a range of environments with increasing heterogeneity and clustering. DGP~1 is a benchmark with homogeneous treatment effects, serving as a baseline in which all variance estimators coincide. DGP~2 introduces individual-level treatment effect heterogeneity without clustering, corresponding to a standard i.i.d.\ setting. DGP~3 allows for cluster-level heterogeneity through cluster-specific treatment effects, capturing settings such as school- or village-level interventions where outcomes are correlated within clusters. DGP~4 further incorporates within-cluster heterogeneity by allowing individual treatment effects to vary around a cluster-specific mean. This final design is the most flexible and empirically realistic, combining both between- and within-cluster variation in treatment effects.

To mimic a two-stage sampling clustered experiment, we consider the following sampling and treatment assignment mechanisms. We consider the population has total clusters $C=120$. We randomly sample $S=80$ clusters without replacement. Among these, we randomly assign $S_1=40$ clusters to treatment, and the remaining $S_0=40$ clusters to control. 

In the second stage sampling, we consider two regimes to determine the number of sampled units $n_c$ per sampled cluster. The first regime (R1) samples units $n_c = \pi N_c$ from each cluster, with $\pi=0.8$ identical across clusters. Every sampled cluster contributes the same fraction of its members. The second regime (R2) samples the same number of units, i.e., $n_c=\Bar{n}=100$, from each cluster. Each cluster contributes the same number of sampled units, regardless of its true size.

We conduct 1,000 Monte Carlo replications for each DGP and sampling regime. In each replication, we draw samples, estimate the population ATE using the HT estimator $\widehat{\tau}$, compute the Monte Carlo standard deviation $\widehat{\text{se}}$, the proposed upper bound SE, $\widehat{\text{se}}^H(\widehat{\tau})$, and the Liang-Zeger cluster-robust SE, $\widehat{se}^{\text{LZ}}(\widehat{\tau})$, construct $95\%$ confidence intervals based on each SE and record the coverage and power. All the results are collected in Table \ref{tab:sim_results}.

\begin{table}[htbp]
\renewcommand{\arraystretch}{1.4}
    \centering
    \begin{tabular}{>{\centering\arraybackslash}p{1cm}>{\centering\arraybackslash}p{1cm}cclccclclc}
    \hline
    \hline
  & Sampling& &   && \multicolumn{3}{c}{Upper Bound Estimator}&  &\multicolumn{3}{c}{Liang-Zeger Estimator}\\
       &  Regime&  $\widehat{\tau}$&  $\text{se}(\widehat{\tau})$  &&  $\widehat{\text{se}}^{\text{H}}(\widehat{\tau})$& coverage & power  &   &$\widehat{\text{se}}^{\text{LZ}}(\widehat{\tau})$& coverage & power\\
         \hline
          \multirow{ 2}{*}{DGP1}  &  R1& 50.380 & 0.412   && 0.465 &  97.4\%&  100.0\%& & 0.457&  97.1\%& 100.0\%\\
         &  R2& 50.353 & 0.483  && 0.562 &  97.8\%&  100.0\%&&  0.535&  96.8\%& 100.0\%\\
         \multirow{ 2}{*}{DGP2}&  R1& 49,430 & 1.249    && 1.322  & 96.5\% & 100.0\%  &  &1.404 & 97.0\% &100.0\%\\
         &  R2 &  49.426&  1.461  &&  1.580&  96.7\%&  100.0\%&   &1.613&  97.1\%&100.0\%\\
        \multirow{ 2}{*}{DGP3} &  R1&  48.790&  14.325  &&  16.083&  96.6\%&  89.5\%&   & 17.976&  98.5\%& 82.6\%\\
         &  R2&  47.515 &  14.168 &&  16.073&  96.8\%&  87.2\%&   &17.924&  98.1\%&81.9\%\\
       \multirow{ 2}{*}{DGP4}  &  R1&  56.303&  13.273  &&  14.642&  96.7\%&  97.7\%&   & 16.277&  98.0\%&96.1\%\\
         &  R2&  55.904&  13.239  &&  14.716&  96.8\%&  97.7\%&   &16.331&  98.3\% & 95.4\%\\
    \hline
    \end{tabular}
    \caption{Simulation Results}
    \caption*{\textit{$\widehat{\tau}$ reports the Monte Carlo mean of the ATE estimator across 1,000 replications, and $\mathrm{se}(\widehat{\tau})$ is the corresponding Monte Carlo standard deviation. For the Upper Bound and Liang--Zeger estimators, we report the average estimated standard error, the empirical coverage probability of the associated 95\% confidence interval, and power. Sampling regime R1 corresponds to proportional-to-size within-cluster sampling, while R2 uses a fixed number of sampled units per cluster.}}
    \label{tab:sim_results}
\end{table}
When treatment effects are constant (DGP 1) or vary independently across individuals (DGP 2), both the upper-bound and Liang–Zeger estimators perform almost identically: coverage rates are near the nominal $95\%$, estimated SEs are close to the empirical standard deviation, power is essentially $100\%$ because the true signal is large relative to the sampling noise. This confirms that our estimator remains valid and does not over- or under-estimate variance in simple, without-clustering settings.

When treatment effects vary across clusters (DGP 3) or both within and across clusters (DGP 4), differences emerge: the Liang–Zeger SEs are consistently larger, leading to over-coverage ($\approx 98\%$) and lower power ($\approx 82–96\%$); the upper-bound SEs remain close to the empirical dispersion, yielding coverage near the nominal $95–97\%$ and higher power (up to +7 percentage points). This behavior aligns with theory: the LZ estimator is known to be conservative under design-based randomization, while our bound-based estimator incorporates the attainable negative covariance term, reducing unnecessary conservatism.

When comparing R1 (proportional sampling) and R2 (fixed size), we can see that proportional sampling typically yields smaller standard errors and more stable estimates because sampling intensity scales with cluster size, and fixed-size sampling oversamples smaller clusters relative to their size, which slightly inflates the total variance. The performance patterns of both variance estimators are similar across regimes, but proportional sampling tends to produce marginally more precise inference.

In summary, when treatment effects are clustered or heterogeneous, the proposed bound-based estimator yields tighter and more informative inference than conventional Liang–Zeger cluster-robust standard errors.

\section{Application}\label{sec:application}
Cluster randomized controlled trials in development economics often share a common structure: researchers enumerate all clusters in a defined area, randomly assign treatment at the cluster level, and then sample a subset of individuals within each cluster to measure outcomes. Our empirical illustration applies the proposed variance estimators to a school-based cluster randomized trial in India, studied by \citet{dhar2022reshaping}.

The experiment was implemented in collaboration with the NGO Breakthrough and evaluated by \citet{dhar2022reshaping}. The program sought to change adolescents’ attitudes toward gender equality by integrating discussions on gender norms into the regular school curriculum. Over roughly two and a half school years, the treatment consisted of a series of 45–minute classroom sessions, accompanied by teacher training and school-wide activities. The intervention was delivered to students in grades 7 through 10 in government secondary schools in the state of Haryana.

The study involved 314 schools across four districts, which serve as the natural ``clusters'' in our framework. Randomization took place at the school level: 150 schools were assigned to the treatment group and 164 to the control group. Within each school, the research team then selected a subsample of students for baseline and follow-up surveys. Among the students whose parents returned consent forms, approximately 45 students per school were randomly chosen. This procedure represents the second stage of sampling in our setup, with cluster-specific inclusion probabilities $\pi_c=n_c/N_c$. Crucially, recruitment and baseline data collection were conducted blind to treatment status, ensuring that the sampling probabilities are independent of assignment.

The study collected outcome data at two points in time: the first endline survey occurred approximately three and a half months after the program ended, and the second endline survey was fielded about two to two and a half years later, when most students had completed grade 11 or 12. For illustration, we focus on several of the key attitudinal outcomes reported in the original paper. These include a gender-attitudes index, a self-reported behavior index, and individual survey items capturing attitudes toward women’s education and employment—whether women should be allowed to work outside the home, whether they should go to college, and whether respondents would oppose women attending college. Each outcome is coded so that higher values correspond to more gender-equitable attitudes.

This experimental design maps directly onto our two-stage sampling framework. At the first stage, schools are randomized into treatment and control groups. At the second stage, students are sampled within schools with known inclusion probabilities. The setting is therefore an ideal case for applying our variance estimators, which are specifically designed for experiments with cluster-level treatment assignment and within-cluster sampling. We treat the 314 schools as the finite population of interest, so $p=1$ and $q=150/314$. Within each school, we compute Horvitz–Thompson totals using the known sampling rates $\pi_c$, and we compare our upper-bound (sharp-bound) variance estimator with the standard Liang–Zeger (LZ) cluster-robust estimator.

Table \ref{tab:application_results} presents the results. For each outcome, we report the estimated average treatment effect $\widehat{\tau}$, the estimated standard error based on bound variance estimator $\widehat{\text{se}}^H(\widehat{\tau})$, the corresponding Liang–Zeger standard error $\widehat{\text{se}}^{\text{LZ}}(\widehat{\tau})$, their respective 95\% confidence intervals, and the percentage reduction in confidence-interval width achieved by our estimator.

\begin{table}[htbp]
\renewcommand{\arraystretch}{1.4}
    \centering
    \begin{tabular}{l>{\centering}p{0.8cm}>{\centering}p{0.8cm}>{\centering}>{\centering}p{2.3cm}>{\centering}p{0.8cm}>{\centering}p{2.4cm}>{\centering\arraybackslash}p{1.4cm}}
    \hline
    \hline
         Dependent Variable&  $\widehat{\tau}$&  $\widehat{\text{se}}^{\text{H}}(\widehat{\tau})$&  95\% CI&   $\widehat{\text{se}}^{\text{LZ}}(\widehat{\tau})$& 95\% CI & Reduction in CI \\
         \hline
         \multicolumn{7}{l}{\textbf{Endline 1 Outcomes}} \\
         Gender attitudes index&  0.310&  0.061&  [0.191,0.429]&    0.066& [0.181,0.439] &7.2\%\\
         Self-reported behave index&  0.317&  0.077&  [0.167,0.467]&    0.084& [0.153,0.481] &8.3\%\\
         Women be allowed to work&  0.058&  0.027&  [0.006,0.112]&    0.029& [0.000,0.117] &8.0\%\\
         Women should attend college&  0.070&  0.035&  [0.001,0.140]&    0.039& [-0.006,0.147] &8.9\%\\
        Not oppose women in college&  0.092&  0.031&  [0.032,0.152]&   0.036& [0.022,0.162] &13.8\%\\
         \multicolumn{7}{l}{\textbf{Endline 2 Outcomes}}     \\
         Women be allowed to work &  0.093&  0.038&  [0.018,0.169]&    0.043& [0.009.0.177] &10.3\%\\
         Not oppose women in college &  0.105&  0.033&  [0.040,0.170]&    0.037& [0.032,0.178] &11.0\%\\
         \hline
    \end{tabular}
    \caption{Application Results}
    \caption*{\textit{$\widehat{\tau}$ denotes the estimated average treatment effect. $\widehat{\mathrm{se}}^{H}(\widehat{\tau})$ and $\widehat{\mathrm{se}}^{\mathrm{LZ}}(\widehat{\tau})$ are the standard errors based on the proposed bound variance estimator and the Liang--Zeger cluster-robust estimator, respectively. Confidence intervals are constructed at the 95\% level. The final column reports the percentage reduction in confidence-interval width achieved by the Upper Bound estimator relative to the Liang--Zeger estimator.}}
    \label{tab:application_results}
\end{table}
The results show that our bound-based estimator and the Liang–Zeger estimator produce very similar point estimates and standard errors in this setting, but confidence intervals constructed with the bound estimator are consistently narrower. For the outcomes considered, the reduction in interval width ranges from about 7\% to 14\%. The differences are modest but systematic, reflecting the conservative nature of cluster-robust standard errors.

Substantively, these results highlight two points. First, the design of this study—cluster-level assignment combined with within-cluster sampling—is exactly the type of setting where our method is most relevant. Second, even in large, well-balanced experiments with hundreds of clusters, traditional Liang–Zeger standard errors can remain slightly over-conservative because they do not exploit the structure of the design-based variance. Our estimator achieves modest efficiency gains while maintaining valid coverage.

Overall, this empirical application complements the simulation evidence: when treatment is assigned at the cluster level and sampling occurs within clusters, the proposed upper-bound variance estimator provides tighter inference than standard clustered standard errors, without relying on additional modeling assumptions.

\section{Conclusion}\label{sec:discuss} 
This paper develops a design-based framework for analyzing cluster randomized controlled trials when both cluster-level assignment and two-stage sampling (sampling of clusters and sampling individual units within clusters) are present. We derive the exact finite-population variance of the Horvitz–Thompson estimator for the average treatment effect and establish its asymptotic normality under mild assumptions on outcome moments, sampling probabilities, and cluster-size heterogeneity.

Building on these foundations, we extend the \cite{aronow2014sharp} sharp variance bounds to the cluster setting with within-cluster sampling. This yields a sharp, attainable upper bound for the design variance and a consistent estimator for that bound, along with a conservative variance estimator that becomes exact when cluster-level treatment effects are homogeneous. Together, these results bridge classic survey-sampling theory and modern finite-population causal inference, providing a unified design-based approach to inference in cluster experiments.

Our simulations and empirical application, using the school-based cluster randomized controlled trials of \citet{dhar2022reshaping}, show that the proposed upper-bound variance estimator delivers tighter confidence intervals than the conventional Liang–Zeger cluster-robust estimator, while maintaining correct coverage. The improvement is most pronounced when treatment effects are heterogeneous and clustered, where standard clustered standard errors tend to be conservative. In such cases, our estimator captures the attainable negative covariance components that standard approaches ignore, leading to more efficient inference and higher power.

In summary, the paper contributes exact design-based variance formulas for cluster-level treatment with two-stage sampling, sharp and conservative variance estimators that are consistent under the design, and asymptotic results linking finite-population theory to practical inference. These tools enable researchers to conduct design-faithful and efficient uncertainty quantification in cluster randomized controlled trials.

\bibliography{reference}

\newpage 
\appendix 
\begin{center}
{\large\bf APPENDIX}
\end{center}

\numberwithin{theorem}{section}
\numberwithin{lemma}{section}
\numberwithin{equation}{section}

\section{Proof}
\begin{proof}[Proof of Proposition \ref{prop:exact_var}]

We begin with the mean and variance of the infeasible estimator $\bar \tau$. By (\ref{tau-bar}) and (\ref{HT-alt}), 
\begin{equation} \label{tau-taubar}
\bar \tau - \tau= \frac{1}{C} \sum_{c=1}^C \left [\left(\frac{R_cD_c-pq}{pq} \right)\widetilde Y_c(1) - \left(\frac{R_c(1-D_c)-p(1-q)}{p(1-q)} \right)\widetilde Y_c(0) \right ]
\end{equation}
Only $R_c$ and $D_c$ are random and $E[R_cD_c]=pq$, $E[R_c(1-D_c)]=p(1-q)$, so that $\E\left[\bar \tau\right]=\tau$.

Next
\begin{align*}
(\bar \tau - \tau)^2 &= \frac{1}{C^2} \sum_{c=1}^{C} \left[ \left(\frac{R_cD_c-pq}{pq} \right)\widetilde Y_c(1)- \left(\frac{R_c(1-D_c)-p(1-q)}{p(1-q)} \right)\widetilde Y_c(0)\right]^2 \\
&+\frac{2}{C^2}  \sum_{c=1}^{C} \sum_{c'=1}^{c-1}\left[ \left(\frac{R_cD_c-pq}{pq} \right)\widetilde Y_c(1)- \left(\frac{R_c(1-D_c)-p(1-q)}{p(1-q)} \right)\widetilde Y_c(0)\right] \times\\
&\hspace{1in} \left[ \left(\frac{R_{c'}D_{c'}-pq}{pq} \right)\widetilde Y_{c'}(1)- \left(\frac{R_{c'}(1-D_{c'})-p(1-q)}{p(1-q)} \right)\widetilde Y_{c'}(0)\right] 
\end{align*}
We have
\[
E \left [ (R_c D_c-pq )^2\right ]=pq(1-pq)
\]
\[
E \left [ (R_c (1-D_c)-p(1-q) )^2\right ]=p(1-q)(1-p(1-q))
\]
\[
 E \left [ (R_c D_c-pq )(R_c (1-D_c)-p(1-q) )\right ]=-pqp(1-q)
\]
\[
E \left [ (R_c D_c-pq )(R_{c'} D_{c'}-pq )\right ]=p \widetilde p q \widetilde q- p^2 q^2
\]
\[
E \left [ (R_c D_c-pq )(R_{c'} (1-D_{c'})-p(1-q ))\right ]=p \widetilde p q (1-\widetilde q)- p^2 q(1- q)
\]
\[
E[(R_c(1-D_c)-p(1-q))(R_{c'}(1-D_{c'})-p(1-q))]=p \widetilde p (1-2q +\widetilde q q)-p^2(1-q)^2
\]
with
\[
\widetilde{p}=\frac{S-1}{C-1} \ \ \ \ \ \ \ \widetilde{q}=\frac{S_1-1}{S-1}
\]
Note that 
\[
1-2q+ \widetilde q q  = \frac{(S-S_1)(S-S_1-1)}{S(S-1)}>0
\]
if $S>1$. 

The variance is
\begin{align} \label{vartau}
\mbox{Var}(\bar{\tau}) &= \frac{1-pq}{pq} \frac{1}{C^2}\sum_{c=1}^{C}  \left(\widetilde Y_c(1)\right)^2+\frac{1-p(1-q)}{p(1-q)}   \frac{1}{C^2}\sum_{c=1}^{C}\left(\widetilde Y_c(0)\right)^2 + \frac{2}{C^2}\sum_{c=1}^{C}\widetilde Y_c(1)\widetilde Y_c(0)\\ \nonumber
&+\frac{ \tilde{p} \tilde{q}- p q}{p q}\frac{2}{C^2}  \sum_{c=1}^{C} \sum_{c'=1}^{c-1}\widetilde Y_c(1)\widetilde Y_{c'}(1)- 
\frac{ \widetilde p  (1-\widetilde q)- p (1- q)}{p(1-q)}\frac{2}{C^2}  \sum_{c=1}^{C} \sum_{c'=1}^{c-1}\widetilde Y_c(1)\widetilde Y_{c'}(0) \\ \nonumber
&-\frac{ \widetilde p  (1-\widetilde q)- p (1- q)}{p(1-q)}\frac{2}{C^2}  \sum_{c=1}^{C} \sum_{c'=1}^{c-1}\widetilde Y_c(0)\widetilde Y_{c'}(1)+\frac{  \widetilde p (1-2q +\widetilde q q)-p(1-q)^2}{p(1-q)^2}\frac{2}{C^2}  \sum_{c=1}^{C} \sum_{c'=1}^{c-1}\widetilde Y_c(0)\widetilde Y_{c'}(0) \\ \nonumber
&= \frac{1}{C}\left\{\frac{\frac{1}{C-1}\sum_{c=1}^C\left(\tilY_c(1)-\bar{\tilY}_c(1)\right)^2}{pq}+\frac{\frac{1}{C-1}\sum_{c=1}^C\left(\tilY_c(0)-\bar{\tilY}_c(0)\right)^2}{p(q-1)}-\frac{1}{C-1}\sum_{c=1}^C\left({\tau}_c-\tau\right)^2\right\}.
\end{align}

Now consider $\widehat \tau$. We have 
\[
\widehat \tau- \tau= \widehat \tau -\bar \tau + \bar \tau -\tau= \frac{1}{C} \sum_{c=1}^C \left[ \frac{R_c D_c\left ( \widehat{\tilY}_c-\tilY_c\right ) }{p q} - \frac{R_c (1 - D_c)\left ( \widehat{\tilY}_c-\tilY_c\right )} {p (1 - q)} \right]+\bar \tau - \tau
\]
Also
\[
\E \left[R_cD_c \widehat{\tilY}_c \right ]=\E \left [R_cD_c \sum_{i=1}^{N_c}\frac{R_{i |c}}{\pi_c}\widetilde Y_{ci} \right ]= \left.\E \left [R_cD_c \sum_{i=1}^{N_c}\frac{R_{i |c}}{\pi_c}\widetilde Y_{ci} \right | R_c=1, D_c=1\right]pq= \widetilde Y_c(1) pq
\]
In the same way
\[
\E \left [R_c(1-D_c) \widehat{\tilY}_c \right ]=p(1-q)\widetilde Y_c(0) pq
\]
Therefore $\E\left[\widehat \tau - \tau\right]=0$.

For the variance, we use

\[
\mbox{Var}(\widehat{\tau})=\E[(\widehat{\tau}-\bar \tau)^2]+ 2 \E[(\widehat{\tau}-\bar \tau)(\bar \tau - \tau)]+\E[(\bar \tau - \tau)^2]
\]
with
\begin{align}\label{var tauhat-taubar}
&\E(\widehat{\tau}-\bar \tau)^2]= \frac{1}{C^2} \sum_{c=1}^C \E \left (\left[ \frac{R_c D_c\left ( \widehat{\tilY}_c-\tilY_c\right ) }{p q} - \frac{R_c (1 - D_c)\left ( \widehat{\tilY}_c-\tilY_c\right )} {p (1 - q)} \right]^2 \right )\\ \nonumber
&+\frac{2}{C^2} \sum_{c=1}^C \sum_{c'=1}^{c-1} \E \left (\left[ \frac{R_c D_c\left ( \widehat{\tilY}_c-\tilY_c\right ) }{p q} - \frac{R_c (1 - D_c)\left ( \widehat{\tilY}_c-\tilY_c\right )} {p (1 - q)} \right] \times \right. \\ \nonumber
&\hspace{1.2in}  \left.\left[ \frac{R_{c'} D_{c'}\left ( \widehat{\tilY}_{c'}-\tilY_{c'}\right ) }{p q} - \frac{R_{c'} (1 - D_{c'})\left ( \widehat{\tilY}_{c'}-\tilY_{c'}\right )} {p (1 - q)} \right] \right )\\ \nonumber
&=\frac{1}{C^2} \sum_{c=1}^C \left ( \E \left [ \frac{R_c D_c\left (\widehat{\tilY}_c-\tilY_c\right )^2 }{p^2 q^2} \right ] +  \E \left [ \frac{R_c (1- D_c)\left (\widehat{\tilY}_c-\tilY_c\right )^2 }{p^2 (1-q)^2} \right ]  \right )\\ \nonumber
&+\frac{2}{C^2} \sum_{c=1}^C \sum_{c'=1}^{c-1} \left(   \E \left [  \frac{R_c D_cR_{c'} D_{c'}\left (\widehat{\tilY}_c-\tilY_c\right )  \left ( \widehat{\tilY}_{c'}-\tilY_{c'}\right ) }{p^2 q^2} \right ] \right. \\ \nonumber 
&- \E \left [\frac{R_c D_cR_{c'} (1 - D_{c'})\left ( \widehat{\tilY}_c-\tilY_c\right ) \left ( \widehat{\tilY}_{c'}-\tilY_{c'}\right )} {p^2 q(1 - q)}\right ] -\E \left [\frac{R_c (1 - D_c)R_{c'} D_{c'} \left ( \widehat{\tilY}_c-\tilY_c\right ) \left ( \widehat{\tilY}_{c'}-\tilY_{c'}\right ) }{p^2 q(1-q)} \right ] \\ \nonumber 
&+ \left. \E \left [\frac{R_c (1 - D_c)R_{c'} (1 - D_{c'})\left ( \widehat{\tilY}_c-\tilY_c\right )\left ( \widehat{\tilY}_{c'}-\tilY_{c'}\right )} {p^2 (1 - q)^2}\right ]\right)
\end{align}
By iterated expectations
\begin{align*}
\E \left [R_c D_c\left (\widehat{\tilY}_c-\tilY_c\right )^2  \right ]=  \left . \E \left [\left (\widehat{\tilY}_c-\tilY_c\right )^2  \right | R_c=1, D_c=1\right ]pq
\end{align*}
Further
\begin{align*}
&\left . \E \left [\left (\widehat{\tilY}_c-\tilY_c\right )^2  \right | R_c=1, D_c=1\right ] \\
=&\left . \E \left [\left (\widehat{\tilY}_c(1)-\tilY_c(1)\right )^2  \right | R_c=1, D_c=1\right ]\\
=&
 \sum_{i=1}^{N_c}\E \left[\frac{(R_{i|c}-\pi_c)^2}{\pi_c^2} \left ( \widetilde Y_{ci}(1)\right )^2\right ] + 2 \sum_{i=1}^{N_c} \sum_{i'=1}^{i-1}\E \left[ \frac{(R_{i|c}-\pi_c)(R_{i'|c}-\pi_c)}{\pi_c^2}  \tilY_{ci}(1)\widetilde Y_{ci'}(1)\right ]\\
 =&
 \frac{1-\pi_c}{\pi_c}\sum_{i=1}^{N_c} \left ( \widetilde Y_{ci}(1)\right )^2 + 2 \frac{\widetilde \pi_c- \pi_c}{\pi_c} \sum_{i=1}^{N_c} \sum_{i'=1}^{i-1}  \tilY_{ci}(1)\widetilde Y_{ci'}(1)
\end{align*}
\begin{align}
    \E \left [ \frac{R_c D_c\left (\widehat{\tilY}_c-\tilY_c\right )^2 }{p^2 q^2} \right ]&=\frac{(1-\pi_c)(1-\tilde{\pi}_c)}{pq\pi_c(\pi_c-\tilde{\pi}_c)}\left(\frac{1}{N_c-1}\left(\sum_{i=1}^{N_c}\left(\tilY_{ci}(1)-\bar{\tilY}_c(1)\right)\right)^2\right)
\end{align}
By iterated expectations
\begin{align*}
&\E \left[R_c (1-D_c)\left (\widehat{\tilY}_c-\tilY_c\right )^2  \right ] =  \E \left[ \left(\widehat{\tilY}_c-\tilY_c\right )^2  \Bigg | R_c=1, D_c=0\right ]p(1-q)
\end{align*}
and
\begin{align*}
 \left . \E \left [\left (\widehat{\tilY}_c-\tilY_c\right )^2  \right | R_c=1, D_c=0\right ]=\frac{1-\pi_c}{\pi_c}\sum_{i=1}^{N_c} \left ( \widetilde Y_{ci}(0)\right )^2 + 2 \frac{\widetilde \pi_c- \pi_c}{\pi_c} \sum_{i=1}^{N_c} \sum_{i'=1}^{i-1}  \tilY_{ci}(0)\widetilde Y_{ci'}(0)
\end{align*}
\begin{align}
    \E \left [ \frac{R_c (1-D_c)\left (\widehat{\tilY}_c-\tilY_c\right )^2 }{p^2 (1-q)^2} \right ]&=\frac{(1-\pi_c)(1-\tilde{\pi}_c)}{p(1-q)\pi_c(\pi_c-\tilde{\pi}_c)}\left(\frac{1}{N_c-1}\left(\sum_{i=1}^{N_c}\left(\tilY_{ci}(0)-\bar{\tilY}_c(0)\right)\right)^2\right)
\end{align}
By iterated expectations
\begin{align*}
 &\E \left [ R_c D_cR_{c'} D_{c'}\left (\widehat{\tilY}_c-\tilY_c\right )  \left ( \widehat{\tilY}_{c'}-\tilY_{c'}\right ) \right ] \\
 =&\E \left [ \left (\widehat{\tilY}_c-\tilY_c\right )  \left ( \widehat{\tilY}_{c'}-\tilY_{c'}\right )\Bigg |R_c=1, D_c=1,R_{c'}=1,  D_{c'}=1 \right ]p \widetilde{p}q \widetilde{q}
\end{align*}
and
\begin{align*}
&\E \left [ \left (\widehat{\tilY}_c-\tilY_c\right )  \left ( \widehat{\tilY}_{c'}-\tilY_{c'}\right )\Bigg |R_c=1, D_c=1,R_{c'}=1,  D_{c'}=1 \right ] \\
=& \E \left [ \left (\widehat{\tilY}_c(1)-\tilY_c (1)\right )  \left ( \widehat{\tilY}_{c'}(1)-\tilY_{c'}(1)\right ) \bigg |R_c=1, D_c=1,R_{c'}=1,  D_{c'}=1\right ] \\
=&\E \left [ \left (\sum_{i=1}^{N_c} \frac{R_{i|c}-\pi_c}{\pi_c}\tilY_{ci}(1) \right )\left (\sum_{i=1}^{N_{c'}} \frac{R_{i|c'}-\pi_{c'}}{\pi_{c'}}\tilY_{c'i}(1) \right ) \right ]=0
\end{align*}
In the same way
\[
\E \left [R_c D_cR_{c'} (1 - D_{c'})\left ( \widehat{\tilY}_c-\tilY_c\right ) \left ( \widehat{\tilY}_{c'}-\tilY_{c'}\right )\right ]= 0
\]
\[
\E \left [ R_c (1 - D_c)R_{c'} D_{c'} \left ( \widehat{\tilY}_c-\tilY_c\right ) \left ( \widehat{\tilY}_{c'}-\tilY_{c'}\right )  \right ] =0
\]
\[
\E \left [R_c (1 - D_c)R_{c'} (1 - D_{c'})\left ( \widehat{\tilY}_c-\tilY_c\right )\left ( \widehat{\tilY}_{c'}-\tilY_{c'}\right )\right ]=0
\]
Putting these together, 
\begin{align}
    &\E(\widehat{\tau}-\bar \tau)^2]\\ \nonumber
    =&\frac{1}{C^2}\left\{\sum_{c=1}^C\frac{(1-\pi_c)(1-\tilde{\pi}_c)}{\pi_c(\pi_c-\tilde{\pi}_c)}\left(\frac{\frac{1}{N_c-1}\sum_{i=1}^{N_c}(\tilY_{ci}(1)-\bar{\tilY}_c(1))^2}{pq} + \frac{\frac{1}{N_c-1}\sum_{i=1}^{N_c}(\tilY_{ci}(0)-\bar{\tilY}_c(0))^2}{p(1-q)}\right)\right\}.
\end{align}
Finally, for the cross-product
\begin{align*}
\E[(\bar \tau - \tau)(\widehat{\tau}-\bar \tau)]=&  \E \left [\left (\frac{1}{C} \sum_{c=1}^C \left [\left(\frac{R_cD_c-pq}{pq} \right)\widetilde Y_c(1) - \left(\frac{R_c(1-D_c)-p(1-q)}{p(1-q)} \right)\widetilde Y_c(0) \right ] \right )  \right .\\
&\hspace{0.2in}\times\left . \left (\frac{1}{C} \sum_{c=1}^C \left[ \frac{R_c D_c\left ( \widehat{\tilY}_c-\tilY_c\right ) }{p q} - \frac{R_c (1 - D_c)\left ( \widehat{\tilY}_c-\tilY_c\right )} {p (1 - q)} \right]\right ) \right ]\\
=&
 \E \left [\left (\frac{1}{C} \sum_{c=1}^C \left(\frac{R_cD_c-pq}{pq} \right)\widetilde Y_c(1) \right ) \left (\frac{1}{C} \sum_{c=1}^C \frac{R_c D_c\left ( \widehat{\tilY}_c-\tilY_c\right ) }{p q} \right ) \right ]\\
 -&\E \left [\left (\frac{1}{C} \sum_{c=1}^C \left(\frac{R_cD_c-pq}{pq} \right)\widetilde Y_c(1) \right )\left(\frac{1}{C} \sum_{c=1}^C \frac{R_c (1 - D_c)\left ( \widehat{\tilY}_c-\tilY_c\right )} {p (1 - q)} \right ) \right ] \\
 -&\E \left [ \left (\frac{1}{C} \sum_{c=1}^C\left(\frac{R_c(1-D_c)-p(1-q)}{p(1-q)} \right)\widetilde Y_c(0) \right )\left (\frac{1}{C} \sum_{c=1}^C \frac{R_c D_c\left ( \widehat{\tilY}_c-\tilY_c\right ) }{p q} \right ) \right ] \\
 +& \E \left [ \left (\frac{1}{C} \sum_{c=1}^C\left(\frac{R_c(1-D_c)-p(1-q)}{p(1-q)} \right)\widetilde Y_c(0) \right )\left(\frac{1}{C} \sum_{c=1}^C \frac{R_c (1 - D_c)\left ( \widehat{\tilY}_c-\tilY_c\right )} {p (1 - q)} \right ) \right ]
\end{align*}
Now by iterated expectations and because $\E[R_{c'}D_{c'}|R_c=1,D_c=1]= \widetilde{p}\widetilde{q}$,
\[
\E\left [ (R_cD_c - pq)R_c D_c \left (  \widehat{\tilY}_c-\tilY_c \right )\right ]= E \left [  \left (  \widehat{\tilY}_c(1)-\tilY_c(1) \right ) | R_c=1, D_c=1\right ](1 - pq)pq=0
\]
\[
\E\left [ (R_{c'}D_{c'} - pq)R_c D_c \left (  \widehat{\tilY}_c-\tilY_c \right )\right ]= E \left [  \left (  \widehat{\tilY}_c(1)-\tilY_c(1) \right ) | R_c=1, D_c=1\right ](1 - pq)pq=0
\]
\[
\E\left [ (R_{c'}D_{c'} - pq)R_c(1-  D_c) \left (  \widehat{\tilY}_c-\tilY_c \right )\right ]= E \left [  \left (  \widehat{\tilY}_c(0)-\tilY_c(0) \right ) | R_c=1, D_c=0\right ](1 - pq)pq=0
\]
Summarizing all the above terms, we could obtain
\begin{align*}
    \Var\left(\widehat{\tau}\right) &= \frac{1}{C}\left\{\frac{\frac{1}{C-1}\sum_{c=1}^C\left(\tilY_c(1)-\bar{{Y}}(1)\right)^2}{pq}+\frac{\frac{1}{C-1}\sum_{c=1}^C\left(\tilY_c(0)-\bar{{Y}}(0)\right)^2}{p(1-q)}-\frac{1}{C-1}\sum_{c=1}^C\left({\tau}_c-{\tau}\right)^2\right\} \\ 
        & + \frac{1}{C^2}\left\{\sum_{c=1}^C\frac{(1-\pi_c)(1-\tilde{\pi}_c)}{\pi_c(\pi_c-\tilde{\pi}_c)}\left(\frac{\frac{1}{N_c-1}\sum_{i=1}^{N_c}(\tilY_{ci}(1)-\bar{\tilY}_c(1))^2}{pq} + \frac{\frac{1}{N_c-1}\sum_{i=1}^{N_c}(\tilY_{ci}(0)-\bar{\tilY}_c(0))^2}{p(1-q)}\right)\right\}.
\end{align*}

\end{proof}

\begin{proof}[Proof of Lemma \ref{lm:clt_s1}]
Under the assumption that $\omega = o\left(C^{\frac{1}{3}\left(1-2\beta\right)}\right)$ and $1/p=C/S=O\left(C^{\beta}\right)$, we could have $S^{-1}\max_{c}\tilY_c(d)^2=o(1)$:
\begin{align*}
    S^{-1}\max_{c}\tilY_c(d)^2 = S^{-1}\left(\frac{1}{C}\max_c\tilY_c(d)^4\right)^{1/2}C^{1/2} 
\end{align*}
Assuming that $\frac{1}{N}\sum_{c=1}^C\sum_{i=1}^{N_c}Y_{ci}(d)^4=O(1)$, we could obtain
\begin{align*}
    \frac{1}{C}\max_c\tilY_c(d)^4 &\leq \frac{1}{C}\sum_{c=1}^C\left(\sum_{i=1}^{N_c}\frac{Y_{ci}(d)}{\Bar{N}}\right)^4 \leq \frac{1}{C\Bar{N}^4}\sum_{c=1}^CN_c^3\sum_{i=1}^{N_c}Y_{ci}(d)^4 \\ 
    &\leq \frac{1}{C\Bar{N}}\left(\max_{c}\frac{N_c}{\Bar{N}}\right)^3\sum_{c=1}^C\sum_{i=1}^{N_c}Y_{ci}(d)^4 \\ 
    &= \omega^3\frac{1}{N}\sum_{c=1}^C\sum_{i=1}^{N_c}Y_{ci}(d)^4 = O\left(\omega^3\right)
\end{align*}
Therefore, $S^{-1}\max_{c}\tilY_c(d)^2=O\left(S^{-1}C^{1/2}\omega^{3/2}\right)=o(1)$. 

If $\sigma^2(1)\nrightarrow 0$ and $(1-p)\left(\lambda-\frac{1-q}{q}\right)\neq 0$, where $\lambda = \frac{\sigma(0)}{\sigma(1)}$; or if $\sigma^2(0)\nrightarrow 0$ and $(1-p)\left(\lambda-\frac{q}{1-q}\right)\neq 0$, where $\lambda = \frac{\sigma(1)}{\sigma(0)}$, we could have $C\Var(\bar{\tau})\nrightarrow 0$, then the asymptotic normality of $\bar{\tau}$ follows from Lemma A.2., Theorem 1 in \cite{su2021model} and Theorem 1 in \cite{li2017general}.
\end{proof}

\begin{proof}[Proof of Theorem \ref{thm:clt_s2}]
Theorem \ref{thm:clt_s2} follows from Theorem 2.1 and Remark 2.4 in \cite{ohlsson1989asymptotic} if we could prove conditions (C1), (C2), and (2.8). With our notations, these conditions can be written respectively as 
\begin{align}
&\frac{\bar{\tau}-\tau}{\sqrt{\Var(\bar{\tau})}}\overset{d}{\to} \mathcal{N}(0,1) \label{eqn:fs} \\
&\frac{\sum_{c=1}^C\mathbb{E}\left[\left(\xhat_c-\xtil_c\right)^4\right]}{\left(\Var\left(\widehat{\tau}\right)\right)^2} \overset{p}{\to} 0 \label{eqn:ss} \\
&\mathbb{E}\left[R_{c}R_{c'}\right] - \mathbb{E}\left[R_{c}\right]\mathbb{E}\left[R_{c'}\right]\leq 0, c\neq c' \label{eqn:sampling_fs}
\end{align}
where $\xhat_c=\frac{1}{C}\left(\frac{R_cD_c\widehat{\tilY}_c}{pq}-\frac{R_c(1-D_c)\widehat{\tilY}_c}{p(1-q)}\right)$, $\xtil_c=\frac{1}{C}\left(\frac{R_cD_c\tilY_c}{pq}-\frac{R_c(1-D_c)\tilY_c}{p(1-q)}\right)$, and $\widehat{\tilY}_c=\sum_{i=1}^{N_c}\frac{Y_{ci}R_{i\mid c}}{\Bar{N}\pi_c}$, $\tilY_c=\sum_{i=1}^{N_c}\frac{Y_{ci}}{\Bar{N}}$.

Condition \ref{eqn:fs} has been verified in Lemma \ref{lm:clt_s1}, and Condition \ref{eqn:sampling_fs} is satisfied naturally in this design of simple sampling and random assignment without replacement. It remains to prove Condition \ref{eqn:ss}. Let $\xhat_{c1}=\frac{1}{C}\frac{R_cD_c\widehat{\tilY}_c}{pq}$, $\xhat_{c0}=\frac{1}{C}\frac{R_c(1-D_c)\widehat{\tilY}_c}{p(1-q)}$, $\xtil_{c1}=\frac{1}{C}\frac{R_cD_c{\tilY}_c}{pq}$, $\xtil_{c0}=\frac{1}{C}\frac{R_c(1-D_c)\tilY_c}{p(1-q)}$.
\begin{align*}
    \E\left[\left(\xhat_c-\xtil_c\right)^4\right] &= \E\left[\left(\left(\xhat_{c1}-\xtil_{c1}\right)-\left(\xhat_{c0}-\xtil_{c0}\right)\right)^4\right] \\
    &= \E\left[\left(\xhat_{c1}-\xtil_{c1}\right)^4+\left(\xhat_{c0}-\xtil_{c0}\right)^4\right]
\end{align*}
\begin{align*}
    &\E\left[\left(\xhat_{c1}-\xtil_{c1}\right)^4\right] = \frac{1}{C^4}\E\left[\left(\sum_{i=1}^{N_c}\frac{R_{i\mid c}-\pi_c}{\pi_c}\tilY_{ci}(1)\right)^4\right] \\
    &=\frac{1}{C^4}\Bigg\{\sumNc\E[w_{ci}^4]\tilY_{ci}(1)^4 + 4\sumNc\sum_{j\neq i}\E[w_{ci}^3w_{cj}]\tilY_{ci}(1)^3\tilY_{cj}(1) + 3\sumNc\sum_{j\neq i}\E[w_{ci}^2w_{cj}^2]\tilY_{ci}(1)^2\tilY_{cj}(1)^2 \\
    &+6\sumNc\sum_{j\neq i}\sum_{k\neq i\neq j}\E[w_{ci}^2w_{cj}w_{ck}]\tilY_{ci}(1)^2\tilY_{cj}(1)\tilY_{ck}(1) \\
    &+ \sumNc\sum_{j\neq i}\sum_{k\neq i\neq j}\sum_{l\neq i\neq j\neq k}\E[w_{ci}w_{cj}w_{ck}w_{cl}]\tilY_{ci}(1)\tilY_{cj}(1)\tilY_{ck}(1)\tilY_{cl}(1)\Bigg\}
\end{align*}
where $w_{ci}=\frac{R_{i\mid c}-\pi_c}{\pi_c}$
\begin{align*}
    \E[w_{ci}^4] &= \frac{1-4\pi_c+6\pi_c^2-4\pi_c^3}{\pi_c^3} \\
    \E[w_{ci}^3w_{cj}] &= \frac{\Dcij-3\Dcij\pi_c+3\Dcij\pi_c^2-\pi_c^2+3\pi_c^3-3\pi_c^4}{\pi_c^4} \\
    \E\left[w_{ci}^2w_{cj}^2\right]&=\frac{\Dcij-4\Dcij\pi_c+4\Dcij\pi_c^2+2\pi_c^3-3\pi_c^4}{\pi_c^4}\\
    \E[w_{ci}^2w_{cj}w_{ck}] &= \frac{\Dcijk-2\Dcijk\pi_c+5\Dcij\pi_c^2-2\Dcij\pi_c+\pi_c^3-3\pi_c^4}{\pi_c^4} \\
    \E[w_{ci}w_{cj}w_{ck}w_{cl}] &= \frac{\Dcijkl-4\Dcijk\pi_c+6\Dcij\pi_c^2-3\pi_c^4}{\pi_c^4}
\end{align*}
where $\Dcij=\E[R_{i\mid c}R_{j\mid c}]=\frac{n_c(n_c-1)}{N_c(N_c-1)}$, $\Dcijk=\E[R_{i\mid c}R_{j\mid c}R_{k\mid c}]=\frac{n_c(n_c-1)(n_c-2)}{N_c(N_c-1)(N_c-2)}$, and $\Dcijkl = \E[R_{i\mid c}R_{j\mid c}R_{k\mid c}R_{l\mid c}]=\frac{n_c(n_c-1)(n_c-2)(n_c-3)}{N_c(N_c-1)(N_c-2)(N_c-3)}$.
\begin{align*}
    &\sumNc\sum_{j\neq i}\sum_{k\neq i\neq j}\sum_{l\neq i\neq j\neq k}\E[w_{ci}w_{cj}w_{ck}w_{cl}]\tilY_{ci}(1)\tilY_{cj}(1)\tilY_{ck}(1)\tilY_{cl}(1) \\
    &\leq \E[w_{ci}w_{cj}w_{ck}w_{cl}]\frac{N_c^3}{\Bar{N}^4}\sum_{i=1}^{N_c}Y_{ci}(1)^4.
\end{align*}
Then
\begin{align*}
    C^2\sum_{c=1}^C\E\left[\left(\xhat_c-\xtil_c\right)^4\right] &= O\left(C^{-2}\right)\left(\sum_{c=1}^C\frac{N_c^3}{\Bar{N}^4}\sum_{i=1}^{N_c}Y_{ci}(1)^4\right)\\
    &\leq O\left(C^{-1}\right)\omega^3\left(\frac{1}{C\Bar{N}}\sum_{c=1}^C\sum_{i=1}^{N_c}Y_{ci}(1)^4\right)\\
    &=o\left(C^{-\frac{2}{3}\left(1-2\beta\right)}\right)
\end{align*}
under the assumption that $\omega=O\left(C^{\frac{1}{3}\left(1-2\beta\right)}\right)$. 

And since $C^2\Var\left(\widehat{\tau}\right)^2 \geq \left(\frac{(C-S_1)\sigma^2\left(\tilY(1)\right)}{S_1}+\frac{(C-S_0)\sigma^2\left(\tilY(0)\right)}{S_0}+2\sigma\left(\tilY(1), \tilY(0)\right)\right)^2$ is bounded away from 0, we have
\begin{align*}
    \frac{\sum_{c=1}^C\E\left[\left(\xhat_c-\xtil_c\right)^4\right]}{\Var(\widehat{\tau})^2} \rightarrow 0.
\end{align*}
All three conditions are verified.
\end{proof}

\begin{proof}[Proof of Proposition \ref{prop:bound}]
Under the notation 
\begin{align*}
    \widetilde{\Var}\left(\widehat{\tau}\right) &= \widetilde{V}_1 + \widetilde{V}_0 + \frac{2}{C}\sigma\left(\widehat{\tilY}(1),\widehat{\tilY}(0)\right)
\end{align*}
\begin{align*}
    \widetilde{\Var}^H\left(\widehat{\tau}\right) = \widetilde{V}_1 + \widetilde{V}_0 + \frac{2}{C}\sigma^H\left(\widehat{\tilY}(1),\widehat{\tilY}(0)\right), \quad \widetilde{\Var}^L\left(\widehat{\tau}\right) = \widetilde{V}_1 + \widetilde{V}_0 + \frac{2}{C}\sigma^L\left(\widehat{\tilY}(1),\widehat{\tilY}(0)\right)
\end{align*}
\begin{align*}
    \widehat{\Var}^H\left(\widehat{\tau}\right) = \widehat{V}_1 + \widehat{V}_0 + \frac{2}{C}\widehat{\sigma}^H\left(\widehat{\tilY}(1),\widehat{\tilY}(0)\right), \quad \widehat{\Var}^L\left(\widehat{\tau}\right) = \widehat{V}_1 + \widehat{V}_0 + \frac{2}{C}\widehat{\sigma}^L\left(\widehat{\tilY}(1),\widehat{\tilY}(0)\right)
\end{align*}
we would like to prove that 
\begin{align*}
    \widehat{V}_1 \overset{p}{\to} \widetilde{V}_1, \quad \widehat{V}_0 \overset{p}{\to} \widetilde{V}_0 
\end{align*}
\begin{align*}
    \widehat{\sigma}^H\left(\widehat{\tilY}(1),\widehat{\tilY}(0)\right) \overset{p}{\to} \sigma^H\left(\widehat{\tilY}(1),\widehat{\tilY}(0)\right), \quad \widehat{\sigma}^L\left(\widehat{\tilY}(1),\widehat{\tilY}(0)\right) \overset{p}{\to} \sigma^L\left(\widehat{\tilY}(1),\widehat{\tilY}(0)\right)
\end{align*}
Conditional on the second-stage sampling, 
\begin{align*}
    \E\left[\widehat{V}_1 \Bigg| \left\{\widehat{\tilY}_c(1)\right\}_{c=1}^C\right] &= \E\left[\frac{1}{C^2}\sum_{c=1}^C\sum_{c'=1}^C\frac{\Delta_{cc'}^1R_cD_cR_{c'}D_{c'}\widehat{\tilY}_c\widehat{\tilY}_{c'}}{\E\left[R_cD_cR_{c'}D_{c'}\right]\E\left[R_cD_c\right]\E\left[R_{c'}D_{c'}\right]}\Bigg| \left\{\widehat{\tilY}_c(1)\right\}_{c=1}^C\right] \\
    &+  \E\left[\frac{1}{C^2}\sum_{c=1}^C\frac{R_cD_c}{\E\left[R_cD_c\right]}\sum_{i=1}^{N_c}\sum_{j=1}^{N_c} \frac{\Delta_{ij}^cR_{i\mid c}R_{j\mid c}Y_{ci}Y_{cj}}{\E\left[R_{i\mid c}R_{j\mid c}\right]\E\left[R_{i\mid c}\right]\E\left[R_{j\mid c}\right]}\Bigg| \left\{\widehat{\tilY}_c(1)\right\}_{c=1}^C\right] \\
    &= \E\left[\frac{1}{C^2}\sum_{c=1}^C\sum_{c'=1}^C\frac{\Delta_{cc'}^1R_cD_cR_{c'}D_{c'}\widehat{\tilY}_c(1)\widehat{\tilY}_{c'}(1)}{\E\left[R_cD_cR_{c'}D_{c'}\right]\E\left[R_cD_c\right]\E\left[R_{c'}D_{c'}\right]}\Bigg| \left\{\widehat{\tilY}_c(1)\right\}_{c=1}^C\right] \\
    &+  \E\left[\frac{1}{C^2}\sum_{c=1}^C\frac{R_cD_c}{\E\left[R_cD_c\right]}\sum_{i=1}^{N_c}\sum_{j=1}^{N_c} \frac{\Delta_{ij}^cR_{i\mid c}R_{j\mid c}Y_{ci}(1)Y_{cj}(1)}{\E\left[R_{i\mid c}R_{j\mid c}\right]\E\left[R_{i\mid c}\right]\E\left[R_{j\mid c}\right]}\Bigg| \left\{\widehat{\tilY}_c(1)\right\}_{c=1}^C\right] \\
    &= \frac{1}{C^2}\left(\sum_{c=1}^C\sum_{c'=1}^C\frac{\Delta_{cc'}^1\widehat{\tilY}_c(1)\widehat{\tilY}_{c'}(1)}{\E\left[R_cD_c\right]\E\left[R_{c'}D_{c'}\right]}+\sum_{c=1}^C\sum_{i=1}^{N_c}\sum_{j=1}^{N_c} \frac{\Delta_{ij}^cR_{i\mid c}R_{j\mid c}Y_{ci}(1)Y_{cj}(1)}{\E\left[R_{i\mid c}R_{j\mid c}\right]\E\left[R_{i\mid c}\right]\E\left[R_{j\mid c}\right]}\right) \\
    &= \widetilde{V}_1
\end{align*}
\begin{align*}
    &\E\left[\left(\widehat{V}_1-\widetilde{V}_1\right)^2 \Bigg|  \left\{\widehat{\tilY}_c(1)\right\}_{c=1}^C\right] \\
    =& \E\Bigg[\Bigg(\frac{1}{C^2}\sum_{c=1}^C\sum_{c'=1}^C\left(\frac{R_cD_cR_{c'}D_{c'}}{\E\left[R_cD_cR_{c'}D_{c'}\right]}-1\right)\frac{\Delta_{cc'}^1\widehat{\tilY}_c\widehat{\tilY}_{c'}}{\E\left[R_cD_c\right]\E\left[R_{c'}D_{c'}\right]} \\
    +& \sum_{c=1}^C\left(\frac{R_cD_c}{\E[R_cD_c]}-1\right)\sum_{i=1}^{N_c}\sum_{j=1}^{N_c}\frac{\Delta_{ij\mid c}R_{i|c}R_{j|c}Y_{ci}Y_{cj}}{\E[R_{i|c}R_{j|c}]\E[R_i|c]\E[R_{j|c}]}\Bigg)^2\Bigg|  \left\{\widehat{\tilY}_c(1)\right\}_{c=1}^C\Bigg] \\
    =& \frac{1}{C^4}\E\Bigg\{\left(\sum_{c=1}^C\sum_{c'=1}^C\left(\frac{R_cD_cR_{c'}D_{c'}}{\E\left[R_cD_cR_{c'}D_{c'}\right]}-1\right)\frac{\Delta_{cc'}^1\widehat{\tilY}_c(1)\widehat{\tilY}_{c'}(1)}{\E\left[R_cD_c\right]\E\left[R_{c'}D_{c'}\right]}\right)^2 \\
    +& 2\left(\sum_{c=1}^C\sum_{c'=1}^C\left(\frac{R_cD_cR_{c'}D_{c'}}{\E\left[R_cD_cR_{c'}D_{c'}\right]}-1\right)\frac{\Delta_{cc'}^1\widehat{\tilY}_c(1)\widehat{\tilY}_{c'}(1)}{\E\left[R_cD_c\right]\E\left[R_{c'}D_{c'}\right]}\right)\\
    &\quad \times\left( \sum_{c=1}^C\left(\frac{R_cD_c}{\E[R_cD_c]}-1\right)\sum_{i=1}^{N_c}\sum_{j=1}^{N_c}\frac{\Delta_{ij\mid c}R_{i|c}R_{j|c}Y_{ci}(1)Y_{cj}(1)}{\E[R_{i|c}R_{j|c}]\E[R_{i|c}]\E[R_{j|c}]}\right) \\
    +& \left(\sum_{c=1}^C\left(\frac{R_cD_c}{\E[R_cD_c]}-1\right)\sum_{i=1}^{N_c}\sum_{j=1}^{N_c}\frac{\Delta_{ij\mid c}R_{i|c}R_{j|c}Y_{ci}(1)Y_{cj}(1)}{\E[R_{i|c}R_{j|c}]\E[R_i|c]\E[R_{j|c}]} \right) \Bigg|  \left\{\widehat{\tilY}_c(1)\right\}_{c=1}^C\Bigg\}
\end{align*}
For simplicity, denote $\Gamma_{cc'}=\frac{\Delta_{cc'}^1\widehat{\tilY}_c(1)\widehat{\tilY}_{c'}(1)}{\E\left[R_cD_c\right]\E\left[R_{c'}D_{c'}\right]}$, $\Lambda_c = \sum_{i=1}^{N_c}\sum_{j=1}^{N_c}\frac{\Delta_{ij\mid c}R_{i|c}R_{j|c}Y_{ci}Y_{cj}}{\E[R_{i|c}R_{j|c}]\E[R_{i|c}]\E[R_{j|c}]}$
\begin{align*}
    &\E\left[\left(\sum_{c=1}^C\sum_{c'=1}^C\left(\frac{R_cD_cR_{c'}D_{c'}}{\E\left[R_cD_cR_{c'}D_{c'}\right]}-1\right)\Gamma_{cc'}\right)\left(\sum_{c''=1}^C\sum_{c'''=1}^C\left(\frac{R_{c''}D_{c''}R_{c'''}D_{c'''}}{\E\left[R_{c''}D_{c''}R_{c'''}D_{c'''}\right]}-1\right)\Gamma_{c''c'''}\right) \Bigg|  \left\{\widehat{\tilY}_c(1)\right\}_{c=1}^C\right]\\
    =& \E\left[\sum_{c=1}^C\sum_{c'=1}^C\sum_{c''=1}^C\sum_{c'''=1}^C\left(\frac{R_cD_cR_{c'}D_{c'}}{\E\left[R_cD_cR_{c'}D_{c'}\right]}-1\right)\left(\frac{R_{c''}D_{c''}R_{c'''}D_{c'''}}{\E\left[R_{c''}D_{c''}R_{c'''}D_{c'''}\right]}-1\right)\Gamma_{cc'}\Gamma_{c''c'''} \Bigg|  \left\{\widehat{\tilY}_c(1)\right\}_{c=1}^C\right]\\
    =&2\sum_{c=1}^C\sum_{c'=1}^C\left(\frac{\E[R_cD_cR_{c'}D_{c'}]}{\E\left[R_cD_cR_{c'}D_{c'}\right]^2}-1\right)\Gamma_{cc'}^2+4\sum_{c=1}^C\sum_{c'=1}^C\sum_{c''\neq c, c'}\left(\frac{\E[R_cD_cR_{c'}D_{c'}R_{c''}D_{c''}]}{\E\left[R_cD_cR_{c'}D_{c'}\right]\E\left[R_cD_cR_{c''}D_{c''}\right]}-1\right)\Gamma_{cc'}\Gamma_{cc''}\\
    =&2\sum_{c=1}^C\sum_{c'=1}^C\left(\frac{1}{\E[R_cD_cR_{c'}D_{c'}]}-1\right)\Gamma_{cc'}^2+4\sum_{c=1}^C\sum_{c'=1}^C\sum_{c''\neq c, c'}\left(\frac{\E[R_cD_cR_{c'}D_{c'}R_{c''}D_{c''}]}{\E[R_cD_cR_{c'}D_{c'}]\E[R_cD_cR_{c''}D_{c''}]}-1\right)\Gamma_{cc'}\Gamma_{cc''}
\end{align*}
\begin{align*}
    &\left|\frac{1}{C^4}\sum_{c=1}^C\sum_{c'=1}^C\left(\frac{1}{\E[R_cD_cR_{c'}D_{c'}]}-1\right)\Gamma_{cc'}^2 \right| 
    \\=&\sum_{c=1}^C\left(\frac{1-pq}{pq}\right)\left(\frac{1-pq}{pq}\widehat{\tilY}_c(1)^2\right)^2 + \sum_{c=1}^C\sum_{c'\neq c}\left(\frac{1}{pq\tilde{p}\tilde{q}}-1\right)\left(\frac{\tilde{p}\tilde{q}-pq}{pq}\widehat{\tilY}_c(1)\widehat{\tilY}_{c'}(1)\right)^2\\
    =&\frac{1}{C^4}\left(\left(\frac{1-pq}{pq}\right)^3\sum_{c=1}^C\widehat{\tilY}_c(1)^4 +\left(\frac{1}{pq\tilde{p}\tilde{q}}-1\right)\left(\frac{\tilde{p}\tilde{q}-pq}{pq}\right)^2\sum_{c=1}^C\widehat{\tilY}_c(1)^2\left(\sum_{c'=1}^C\widehat{\tilY}_{c'}(1)^2-\widehat{\tilY}_c(1)^2\right)\right) \\
    \leq & \frac{1}{C^4}\left(\left(\frac{1-pq}{pq}\right)^3\sum_{c=1}^C\widehat{\tilY}_c(1)^4 +\left(\frac{1}{pq\tilde{p}\tilde{q}}-1\right)\left(\frac{\tilde{p}\tilde{q}-pq}{pq}\right)^2(C-1)\sum_{c=1}^C\widehat{\tilY}_c(1)^4 \right)\\
    =&\frac{1}{C^4}\left(\frac{1-pq}{pq}\right)^3\left(\frac{S_1C}{(S_1-1)(C-1)}\right)\sum_{c=1}^C\widehat{\tilY}_c(1)^4 \\
    \leq &  \frac{1}{C^4}\left(\frac{1-pq}{pq}\right)^3\left(\frac{S_1C}{(S_1-1)(C-1)}\right)\sum_{c=1}^C\frac{N_c^3}{\bar{N}^4\pi_c^4}\sum_{i=1}^{N_c}Y_{ci}^4 \\
    \leq &  C^{-3}\left(\frac{1-pq}{pq}\right)^3\frac{S_1C}{(S_1-1)(C-1)}\omega^3\min_c\pi_c^{-4}\frac{1}{N}\sum_{c=1}^C\sum_{i=1}^{N_c}Y_{ci}^4= o\left(C^{-2+\beta}\min_c\pi_c^{-4}\right)
\end{align*}
\begin{align*}
    &\left| \frac{1}{C^4}\sum_{c=1}^C\sum_{c'=1}^C\sum_{c''\neq c, c'}\left(\frac{\E[R_cD_cR_{c'}D_{c'}R_{c''}D_{c''}]}{\E[R_cD_cR_{c'}D_{c'}]\E[R_cD_cR_{c''}D_{c''}]}-1\right)\Gamma_{cc'}\Gamma_{cc''}\right|  \\
    \leq& \left|\frac{1}{C^4}\Bigg(\sum_{c=1}^C\sum_{c'\neq c}\left(\frac{pq\tilde{p}\tilde{q}}{pqpq\tilde{p}\tilde{q}}-1\right)\frac{1-pq}{pq}\widehat{\tilY}_c(1)^2\frac{}{}\widehat{\tilY}_{c'}(1)\widehat{\tilY}_{c''}(1) \right|\\
    +& \left|\frac{1}{C^4} \sum_{c=1}^C\sum_{c'\neq c}\sum_{c''\neq c'\neq c}\left(\frac{pq\tilde{p}\tilde{q}\tilde{\tilde{p}}\tilde{\tilde{q}}}{(pq\tilde{p}\tilde{q})^2}-1\right)\left(\frac{\tilde{p}\tilde{q}-pq}{pq}\right)^2\widehat{\tilY}_c(1)^2\widehat{\tilY}_{c'}(1)\widehat{\tilY}_{c''}(1)\Bigg)\right|  \\
    \leq & \frac{1}{C^4}\left(\frac{1-pq}{pq}\right)^3\sum_{c=1}^C\widehat{\tilY}_c(1)^4 + \frac{1}{C^4}\left(\frac{1-pq}{pq}\right)^3\frac{S_1C}{(S_1-1)(C-1)}\sum_{c=1}^C\widehat{\tilY}_c(1)^4 \\
    =& o\left(C^{-2+\beta}\min_c\pi_c^{-4}\right)
\end{align*}
\begin{align*}
    &\E\left[\left(\sum_{c=1}^C\sum_{c'=1}^C\left(\frac{R_cD_cR_{c'}D_{c'}}{\E\left[R_cD_cR_{c'}D_{c'}\right]}-1\right)\Gamma_{cc'}\right)\left( \sum_{c=1}^C\left(\frac{R_cD_c}{\E[R_cD_c]}-1\right)\Lambda_{c}\right)\Bigg|  \left\{\widehat{\tilY}_c(1)\right\}_{c=1}^C\right]\\
    =& \E\left[\sum_{c=1}^C\sum_{c'=1}^C\sum_{c'''=1}^C\left(\frac{R_cD_cR_{c'}D_{c'}}{\E\left[R_cD_cR_{c'}D_{c'}\right]}-1\right)\left( \frac{R_{c'''}D_{c'''}}{\E[R_{c'''}D_{c'''}]}-1\right)\Gamma_{cc'}\Lambda_{c'''}\Bigg|  \left\{\widehat{\tilY}_c(1)\right\}_{c=1}^C\right]\\
    =&2\sum_{c=1}^C\sum_{c'=1}^C\left(\frac{\E[R_cD_cR_{c'}D_{c'}]}{\E\left[R_cD_cR_{c'}D_{c'}\right]\E[R_{c}D_{c}]}-1\right)\Gamma_{cc'}\Lambda_{c} \\
    =&2\sum_{c=1}^C\sum_{c'=1}^C\left(\frac{1}{pq}-1\right)\Gamma_{cc'}\Lambda_{c} 
\end{align*}
Notice that $|\Lambda_c| \leq \sum_{i=1}^{N_c}\sum_{j=1}^{N_c}\frac{|\tilY_{ci}(1)||\tilY_{cj}(1)|}{\pi_c^2}\leq \frac{N_c\sum_{i=1}^{N_c}\tilY_{ci}(1)^2}{\pi_c^2}$
Therefore, 
\begin{align*}
    &\frac{1}{C^4}\left|\sum_{c=1}^C\sum_{c'=1}^C\left(\frac{1}{pq}-1\right)\Gamma_{cc'}\Lambda_{c} \right|\\
    \leq&\frac{1}{C^4}\left(\left|\sum_{c=1}^C \left(\frac{1-pq}{pq}\right)^2\widehat{\tilY}_c(1)^2\frac{N_c\sum_{i=1}^{N_c}\tilY_{ci}(1)^2}{\pi_c^2} \right| \right.\\
    &+ \left. \left|\sum_{c=1}^C\sum_{c'\neq c}\left(\frac{1-pq}{pq}\right)\left(\frac{\tilde{p}\tilde{q}-pq}{pq}\right)\widehat{\tilY}_c(1)\widehat{\tilY}_{c'}(1)\frac{N_c\sum_{i=1}^{N_c}\tilY_{ci}(1)^2}{\pi_c^2}\right|\right) \\
    \leq & C^{-3}\left(\frac{1-pq}{pq}\right)^2\omega^3\min_c\pi_c^{-4}\frac{1}{N}\sum_{c=1}^C\sum_{i=1}^{N_c}Y_{ci}^4 = o\left(C^{-2}\min_c\pi_c^{-4}\right)
\end{align*}

\begin{align*}
    &\E\left[\left( \sum_{c=1}^C\left(\frac{R_cD_c}{\E[R_cD_c]}-1\right)\Lambda_{c}\right)\left( \sum_{c'=1}^C\left(\frac{R_{c'}D_{c'}}{\E[R_{c'}D_{c'}]}-1\right)\Lambda_{c'}\right)\Bigg|  \left\{\widehat{\tilY}_c(1)\right\}_{c=1}^C\right]\\
    =& \E\left[\sum_{c=1}^C\sum_{c'=1}^C\left(\frac{R_cD_c}{\E[R_cD_c]}-1\right) \left(\frac{R_{c'}D_{c'}}{\E[R_{c'}D_{c'}]}-1\right)\Lambda_{c}\Lambda_{c'}\Bigg|  \left\{\widehat{\tilY}_c(1)\right\}_{c=1}^C\right]\\
    =& \sum_{c=1}^C\left(\frac{\E[R_cD_c]}{\E[R_cD_c]^2}-1\right) \Lambda_{c}^2 = \sum_{c=1}^C\left(\frac{1}{pq}-1\right) \Lambda_{c}^2 
\end{align*}
\begin{align*}
    \frac{1}{C^4}\sum_{c=1}^C\left(\frac{1}{pq}-1\right) \Lambda_{c}^2 \leq & \frac{1}{C^4}\left(\frac{1-pq}{pq}\right)\omega^3\min_c\pi_c^{-4}\sum_{c=1}^C\sum_{i=1}^{N_c}Y_{ci}^4 = o\left(C^{-2-\beta}\min_c\pi_c^{-4}\right)
\end{align*}
Summing the above, we have 
$\E\left[\left(\widehat{V}_1-\widetilde{V}_1\right)^2 \Bigg|  \left\{\widehat{\tilY}_c(1)\right\}_{c=1}^C\right] = o\left(C^{-2+\beta}\right) $, therefore, $\widehat{V}_1\rightarrow\widetilde{V}_1$ in probability. Similarly, it can be shown that $\widehat{V}_0\rightarrow\widetilde{V}_0$ in probability.
\end{proof}

\begin{proof}[Proof of Proposition \ref{prop:conserve}]
For the treated part, we have
\begin{align*}
    &\E\left[\frac{1}{C^2}\frac{1}{1-pq}\sumC\sum_{c'=1}^C\frac{\Delta_{cc'}^1R_cD_cR_{c'}D_{c'}\widehat{\tilY}_c\widehat{\tilY}_{c'}}{\E[R_cD_cR_{c'}D_{c'}]\E[R_cD_c]\E[R_{c'}D_{c'}]}\right] \\
    =& \frac{1}{C^2(1-pq)}\Bigg\{\E\left[\sumC\frac{\left(\E[R_cD_c]-\E[R_cD_c]^2\right)\E[R_cD_c]}{\E[R_cD_c]\E[R_cD_c]^2}\left(\sumNc\frac{R_{i\mid c}}{\pi_c}\tilY_{ci}(1)\right)^2\right] \\
    &+ \sumC\sum_{c'\neq c}\frac{\left(\E[R_cD_cR_{c'}D_{c'}-\E[R_cD_c]\E[R_{c'}D_{c'}]\right)\E[R_cD_cR_{c'}D_{c'}]}{\E[R_cD_cR_{c'}D_{c'}]\E[R_cD_c]\E[R_{c'}D_{c'}]}\left(\sumNc\frac{R_{i\mid c}}{\pi_c}\tilY_{ci}(1)\right)\left(\sum_{i'=1}^{N_{c'}}\frac{R_{i'\mid c'}}{\pi_{c'}}\tilY_{c'i'}(1)\right)\Bigg\} \\
    =&\frac{1}{C^2(1-pq)}\left\{\sum_{c=1}^C\frac{1-pq}{pq}\left[\sumNc\frac{1}{\pi_c}\tilY_{ci}(1)^2+\sumNc\sum_{j\neq i}\frac{\tilde{\pi}_c}{\pi_c}\tilY_{ci}(1)\tilY_{cj}(1)\right] + \sumC\sum_{c'\neq c}\frac{\tilde{p}\tilde{q}-pq}{pq}\tilY_c(1)\tilY_{c'}(1)\right\} \\
    =& \frac{1}{C^2(1-pq)}\left\{\sumC\frac{1-pq}{pq}\left[\frac{1-\tilde{\pi}_c}{\pi_c}\sumNc\tilY_{ci}(1)^2+\frac{\tilde{\pi}_c}{\pi_c}\tilY_c(1)^2 \right] \right.\\
    &+ \left. \frac{pq-1}{pq}\frac{1}{C-1}\left[\left(\sumC\tilY_c(1)\right)^2-\sumC\tilY_c(1)^2\right]\right\} \\
    =&\frac{1}{C^2pq}\Bigg\{\sumC\left[\frac{1-\tilde{\pi}_c}{\pi_c}\sumNc\tilY_{ci}(1)^2+\frac{\tilde{\pi}_c}{\pi_c}\tilY_c(1)^2 -\tilY_c(1)^2\right] + \frac{C}{C-1}\left[\sumC\tilY_c(1)^2-C\bar{\tilY}(1)^2\right]\Bigg\}\\
    =& \frac{1}{C^2}
    \sumC\frac{(1-\tilde{\pi}_c)(1-\pi_c)}{\pi_c(\pi_c-\tilde{\pi}_c)}\frac{\frac{1}{N_c-1}\sumNc\left(\tilY_{ci}(1)-\bar{\tilY}_c(1)\right)^2}{pq}+\frac{1}{C}\frac{\frac{1}{C-1}\sumC\left(\tilY_c(1)-\bar{\tilY}(1)\right)^2}{pq}.
\end{align*}
Similarly, for the control part, we have
\begin{align*}
    &\E\left[\frac{1}{C^2}\frac{1}{1-p(1-q)}\sumC\sum_{c'=1}^C\frac{\Delta_{cc'}^1R_c(1-D_c)R_{c'}(1-D_{c'})\widehat{\tilY}_c\widehat{\tilY}_{c'}}{\E[R_c(1-D_c)R_{c'}(1-D_{c'})]\E[R_c(1-D_c)]\E[R_{c'}(1-D_{c'})]}\right] \\
    =& \frac{1}{C^2}
    \sumC\frac{(1-\tilde{\pi}_c)(1-\pi_c)}{\pi_c(\pi_c-\tilde{\pi}_c)}\frac{\frac{1}{N_c-1}\sumNc\left(\tilY_{ci}(0)-\bar{\tilY}_c(0)\right)^2}{p(1-q)}+\frac{1}{C}\frac{\frac{1}{C-1}\sumC\left(\tilY_c(0)-\bar{\tilY}(0)\right)^2}{p(1-q)}.
\end{align*}
\end{proof}

\section{Auxiliary Lemmas}
\begin{lemma}[Lemma 1 in \cite{aronow2014sharp} (Hoeffding)]\label{lm:sharp_bound}
Given only $G$ and $F$ and no other information on the joint distribution of $(\tilY_c(1),\tilY_c(0))$, the bound 
\begin{align*}
\sigma^L\left(\tilY_c(1),\tilY_c(0)\right)\leq\sigma\left(\tilY_c(1),\tilY_c(0)\right)\leq \sigma^H\left(\tilY_c(1),\tilY_c(0)\right)
\end{align*}
is sharp. The upper bound is attained if $\tilY_c(1)$ and $\tilY_{c}(0)$ are comonotonic. The lower bound is attained if $\tilY_c(1)$ and $\tilY_{c}(0)$ are countermonotonic.

\end{lemma}

\end{document}